\documentclass [11pt] {article}

 \usepackage{fullpage}

\usepackage{graphics}
\usepackage[dvips]{epsfig}

\usepackage{amsmath}
\usepackage{amssymb}
\usepackage{amsfonts}
\usepackage{graphicx}

\usepackage{cite}

\usepackage{array}
\usepackage{color}
\usepackage{algorithm}
\usepackage{algorithmic}

\usepackage[english]{babel}

\usepackage{caption,subcaption}
\usepackage{multirow}
\usepackage{rotating}

\begin{document}

\newtheorem{theorem}{Theorem}[section]
\newtheorem{lemma}{Lemma}[section]
\newtheorem{corollary}{Corollary}[section]
\newtheorem{claim}{Claim}[section]
\newtheorem{proposition}{Proposition}[section]
\newtheorem{definition}{Definition}[section]
\newtheorem{fact}{Fact}[section]
\newtheorem{example}{Example}[section]

\newcommand{\cA}{{\cal A}}
\newcommand{\cB}{{\cal B}}
\newcommand{\cC}{{\cal C}}
\newcommand{\cZ}{{\cal Z}}
\newcommand{\cG}{{\cal G}}
\newcommand{\cN}{{\cal N}}
\newcommand{\cU}{{\cal U}}
\newcommand{\cT}{{\cal T}}
\newcommand{\cS}{{\cal S}}
\newcommand{\cP}{{\cal P}}
\newcommand{\cL}{{\cal L}}
\newcommand{\cV}{{\cal V}}
\newcommand{\loc}{{\cal LOCAL}}
\newcommand{\cY}{{\cal Y}}
\newcommand{\ai}{\alpha_i}
\newcommand{\bi}{\beta_i}
\newcommand{\gi}{\gamma_i}
\newcommand{\di}{\delta_i}

\newcommand{\oai}{\overline{\alpha}_i}
\newcommand{\obi}{\overline{\beta}_i}
\newcommand{\ogi}{\overline{\gamma}_i}
\newcommand{\odi}{\overline{\delta}_i}

\newcommand{\qed}{\hfill $\square$ \smallbreak}
\newenvironment{proof}{\noindent{\bf Proof:}}{\qed}

\newcommand{\algo}[1]{
\medskip
\noindent \textbf{Algorithm {\tt #1}}\\
\nopagebreak}

\newcommand{\procedure}[1]{
\medskip
\noindent \textbf{Procedure {\tt #1}}
\nopagebreak \\}

\newcommand{\procend}{\hfill $\diamond$\medskip}

\newcommand{\oddRepair}{{\tt Odd\-Repair}}
\newcommand{\Deactivate}{{\tt Deactivate}}
\newcommand{\evenARepair}{{\tt Even\-Al\-most\-Re\-pair}}
\newcommand{\ringThree}{{\tt Ring\-Three\-Co\-lo\-ring}}
\newcommand{\ringLearning}{{\tt Ring\-Lear\-ning}}
\newcommand{\Elect}{{\tt Elect}}



\def\thefootnote{\fnsymbol{footnote}}

\title{{\bf Short Labeling Schemes for \\
 Topology Recognition in Wireless Tree Networks}}

\author{Barun Gorain\footnotemark[1]
\and Andrzej Pelc\footnotemark[2]
}

\footnotetext[1]{D\'epartement d'informatique, Universit\'e du Qu\'ebec en Outaouais, Gatineau,
Qu\'ebec J8X 3X7, Canada. {\tt baruniitg123@gmail.com}}

\footnotetext[2]{
 D\'epartement d'informatique, Universit\'e du Qu\'ebec en Outaouais, Gatineau,
Qu\'ebec J8X 3X7, Canada. {\tt pelc@uqo.ca}. Partially supported by NSERC discovery grant 8136--2013
and by the Research Chair in Distributed Computing at the
Universit\'e du Qu\'ebec en Outaouais.}

\maketitle

\thispagestyle{empty}

\begin{abstract}
We consider the problem of topology recognition in wireless (radio) networks modeled as undirected graphs. Topology recognition is a fundamental task in which every
node of the network has to output a map of the underlying graph i.e., an isomorphic copy of it,  and situate itself in this map. In wireless networks, nodes communicate in synchronous rounds. In each round a node can either transmit a message to all its neighbors, or stay silent and listen. At the receiving end, a node $v$ hears a message from a neighbor $w$ in a given round, if $v$ listens in this round, and if $w$ is its only neighbor that transmits in this round. Nodes have  labels which are (not necessarily different)
binary strings. The length of a labeling scheme is the largest length of a label.
We concentrate on wireless networks modeled by trees, and we investigate two problems.
\begin{itemize}
\item
What is the shortest labeling scheme that permits topology recognition in all wireless tree networks of diameter $D$ and maximum degree $\Delta$?
\item
What is the fastest topology  recognition algorithm working for all wireless tree networks of diameter $D$ and maximum degree $\Delta$, using such a short labeling scheme?
\end{itemize}
We are interested in deterministic topology recognition algorithms.
For the first problem, we show that the minimum length of a labeling scheme allowing  topology recognition in all trees of maximum degree $\Delta \geq 3$ is
$\Theta(\log\log \Delta)$. For such short schemes, used by an algorithm working for the class of trees of diameter $D\geq 4$ and maximum degree $\Delta \geq 3$, we show
almost matching bounds on the time of topology recognition: an upper bound $O(D\Delta)$, and a lower bound $\Omega(D\Delta^{\epsilon})$, for any constant $\epsilon<1$.

Our upper bounds are proven by constructing a topology recognition algorithm using a labeling scheme of length $O(\log\log \Delta)$ and using time $O(D\Delta)$.
Our lower bounds are proven by constructing a class of trees for which any topology recognition algorithm must use a labeling scheme of length at least $\Omega(\log \log\Delta)$, and a class of trees for which any topology recognition algorithm using a labeling scheme of length $O(\log\log \Delta)$ must use time at least $\Omega(D\Delta^{\epsilon})$, on some tree of this class.

\vspace*{0.5cm}

\noindent
{\bf keywords:} topology recognition, wireless network, labeling scheme, feasibility, tree, time

\vspace*{0.5cm}
\end{abstract}

\pagebreak

\section{Introduction}

\subsection{The model and the problem}

Learning the topology of an unknown network by its nodes is a fundamental distributed task in networks.
Every node of the network has to output a map of the underlying graph, i.e., an isomorphic copy of it,  and situate itself in this map.
Topology recognition can be considered as a preprocessing procedure to many other distributed algorithms which require the knowledge of important
parameters of the network, such as its size, diameter or maximum degree. It can also help to determine the feasibility of some tasks that
depend, e.g., on symmetries existing in the network.

We consider wireless networks, also known as radio networks. Such a network is modeled as a simple undirected connected graph $G=(V,E)$. As it is usually assumed in the algorithmic theory of radio networks \cite{CGR,GPPR,GPX},  all nodes start simultaneously and communicate in synchronous rounds. In each round, a node can either transmit a message to all its neighbors, or stay silent and listen. At the receiving end, a node $v$ hears a message from a neighbor $w$ in a given round, if $v$ listens in this round, and if $w$ is its only neighbor that transmits in this round. We do not assume collision detection: if more than one neighbor of a node $v$ transmits in a given round, node $v$ does not hear anything
(except the background noise that it also hears when no neighbor transmits).

In this paper, we restrict attention to wireless networks modeled by trees, and we are interested in deterministic topology recognition algorithms.
Topology recognition is formally defined as follows. Every node $v$ of a tree $T$ must output a tree $T'$ and a node $v'$ in this tree, such that there exists
an isomorphism $f$ from $T$ to $T'$, for which $f(v)=v'$.
Topology recognition is impossible, if nodes do not have any a priori assigned labels, because then any deterministic algorithm forces all nodes
to transmit in the same rounds, and no communication is possible.
Hence we consider labeled networks. A {\em labeling scheme} for a network represented by a tree $T=(V,E)$ is any function $\cL$ from the set $V$ of nodes into the set $S$ of finite binary strings. The string $\cL(v)$ is called the label of the node $v$.
Note that labels assigned by a labeling scheme are not necessarily distinct. The {\em length} of a labeling scheme $\cL$ is the maximum length of any label assigned by it.

We investigate two problems.
\begin{itemize}
\item
What is the shortest labeling scheme that permits topology recognition in all wireless tree networks of diameter $D$ and maximum degree $\Delta$?
\item
What is the fastest topology  recognition algorithm working for all wireless tree networks of diameter $D$ and maximum degree $\Delta$, using such a short labeling scheme?
\end{itemize}

\subsection{Our results}

For the first problem, we show that the minimum length of a labeling scheme allowing  topology recognition in all trees of maximum degree $\Delta \geq 3$ is
$\Theta(\log\log \Delta)$. For such short schemes, used by an algorithm working for the class of trees of diameter $D\geq 4$ and maximum degree $\Delta \geq 3$, we show
almost matching bounds on the time of topology recognition: an upper bound $O(D\Delta)$, and a lower bound $\Omega(D\Delta^{\epsilon})$, for any constant $\epsilon<1$.

Our upper bounds are proven by constructing a topology recognition algorithm using a labeling scheme of length $O(\log\log \Delta)$ and using time $O(D\Delta)$.
Our lower bounds are proven by constructing a class of trees for which any topology recognition algorithm must use a labeling scheme of length at least $\Omega(\log \log\Delta)$, and a class of trees for which any topology recognition algorithm using a labeling scheme of length $O(\log\log \Delta)$ must use time at least $\Omega(D\Delta^{\epsilon})$, on some tree of this class.

These main results are complemented by establishing complete answers to both problems for very small values of $D$ or $\Delta$.
For trees of  diameter $D=3$ and maximum degree $\Delta \geq 3$, the fastest topology recognition algorithm using a shortest possible scheme (of length
$\Theta(\log\log \Delta)$)  works in time $\Theta(\frac{\log \Delta}{\log\log \Delta})$. The same holds for trees of diameter $D=2$ and maximum degree {\em at most} $\Delta$,
for $\Delta \geq 3$.
Finally, if $\Delta=2$, i.e., for the class of lines,
the shortest labeling scheme permitting topology recognition is of constant length, and the best time of topology recognition using such a scheme for lines
of diameter (length) at most $D$ is $\Theta(\log D)$.

Our results should be contrasted with those from \cite{FPP}, where topology recognition was studied in a different model. The authors of \cite{FPP} considered wired networks in which there are port numbers at each node, and communication proceeds according to the $\cal{LOCAL}$ model \cite{Pe}, where in each round neighbors can exchange all available information without collisions. In this model, they showed a simple topology recognition algorithm working for a labeling scheme of length 1 in time $O(D)$. Thus there was no issue of optimality: both the length of the labeling scheme and the topology recognition time for such a scheme were trivially optimal.
Hence the authors focused on tradeoffs between the length of (longer) schemes and the time of topology recognition.
In our scenario of wireless networks, the labeling schemes must be longer and algorithms for such schemes must be slower, in order to overcome  collisions.

\subsection{Related work}

Algorithmic problems in radio networks modeled as graphs were studied for such tasks as broadcasting \cite{CGR,GPX}, gossiping \cite{CGR,GPPR} and leader election
\cite{KP}. In some cases \cite{CGR,GPPR} the topology of the network was unknown, in others \cite{GPX} nodes were assumed to have a labeled map of the network and could situate themselves in it.

Providing nodes of a network or mobile agents circulating in it with information of arbitrary type (in the form of binary strings) that can be used to perform network tasks more efficiently has been
proposed in \cite{AKM01,CFIKP,DP,EFKR,FGIP,FIP1,FIP2,FKL,FP,GPPR02,IKP,KKKP02,KKP05,SN}. This approach was referred to as
algorithms using {\em informative labeling schemes}, or equivalently, algorithms with {\em advice}.
When advice is given to nodes,  two variations are considered: either the binary string given to nodes is the same for all of them \cite{GMP} or different strings may be given to different nodes
\cite{FKL,FPP}, as in the case
of the present paper. If strings may be different, they can be considered as labels assigned to nodes.
Several authors studied the minimum size of advice (length of labels) required to solve the
respective network problem in an efficient way. The framework of advice or labeling schemes permits to quantify the amount of information
that nodes need for an efficient solution of a given network problem, regardless of the type of information that is provided.

In \cite{CFIKP} the authors investigated the minimum size of advice that has to be given to nodes
to permit graph exploration by a robot.
 In \cite{KKP05}, given a distributed representation of a solution for a problem,
the authors investigated the number of bits of communication needed to verify the legality of the represented solution.
In \cite{FIP1} the authors compared the minimum size of advice required to
solve two information dissemination problems, using a linear number of messages. In \cite{FIP2} the authors
established the size of advice needed to break competitive ratio 2 of an exploration algorithm in trees.
In \cite{FKL} it was shown that advice of constant size permits to carry on the distributed construction of a minimum
spanning tree in logarithmic time.
In \cite{GPPR} short labeling schemes were constructed with the aim to answer queries about the distance between any pair of nodes.
In \cite{EFKR} the advice paradigm was used for online problems.
In the case of \cite{SN} the issue was not efficiency but feasibility: it
was shown that $\Theta(n\log n)$ is the minimum size of advice
required to perform monotone connected graph clearing.
In \cite{IKP} the authors studied radio networks for
which it is possible to perform centralized broadcasting in constant time. They proved that
$O(n)$ bits of advice allow to obtain constant time in such networks, while
$o(n)$ bits are not enough. This is the only paper studying the size of advice in the context of radio networks.
In \cite{FPP} the authors studied the task of topology recognition in wired networks with port numbers.
The differences between this scenario and our setting of radio networks, in the context of topology recognition, was discussed in the previous section.

\section{Preliminaries and organization}

Throughout the paper, $D$ denotes the diameter of the tree and $\Delta$ denotes its maximum degree.
The problem of topology recognition is non-trivial only for $D, \Delta \geq 2$, hence we make this assumption from now on.

According to the definition of labeling schemes, a label of any node should be a finite binary string. For ease of comprehension, we present our labels in a more structured way, as either finite sequences of binary strings, or pairs of such sequences, where each of the component binary strings is later used in the topology recognition algorithm in a particular way. It is well known that a sequence $(s_1,\dots ,s_k)$ of binary strings or a pair $(\sigma_1,\sigma_2)$ of such sequences can be unambiguously coded as a single binary string whose length is a constant multiple of the sum of lengths of all binary strings $s_i$ that compose it. Hence, presenting labels in this more structured way and skipping the details of the encoding does not change the order of magnitude
of the length of the constructed labeling schemes.

Let $T$ be any rooted tree with root $r$,  and let $L(T)$  be a labeling scheme for this tree.
We say that a node $u$  in $T$ {\em reaches} $r$ within time $\tau$ using algorithm $\cA$ if there exists a simple path $u=u_0,u_1,\cdots,u_{k-1},u_k=r$ and a sequence of integers $t_0<t_1<\cdots< t_{k-1}\leq \tau$, such that in round $t_i$, the node $u_i$ is the only child of its parent $u_{i+1}$ that transmits and the node $u_{i+1}$ does not transmit in round $t_i$, according to algorithm $\cA$.

We define the history $H(\cA,\tau)$ of the root $r$ of the tree $T$ as the labeled subtree of $T$ which is spanned by all the nodes that
reach $r$ within time $\tau$, using algorithm $\cA$. The history $H(\cA,\tau)$ is the total information that node $r$ can learn about the tree $T$ in time $\tau$,
using algorithm $\cA$.

The remainder of the paper is organized as follows. In Section \ref{s1}, we present the lower bound
$\Omega(\log\log \Delta)$ on the length of labeling schemes that permit topology recognition for all trees with maximum degree $\Delta \geq 3$. Section \ref{s2}
is devoted to our main results concerning the time of topology recognition using labeling schemes of length $\Theta(\log\log \Delta)$ for trees of maximum degree
$\Delta \geq 3$ and diameter $D\geq 4$. We prove the lower bound
$\Omega (D\Delta^{\epsilon})$ on the time of such schemes, for any constant $\epsilon>0$, and we construct an algorithm using a labeling scheme of length $\Theta(\log\log \Delta)$ and working in time $O(D\Delta)$. In Section \ref{s3}, we give the solution to both our problems for the remaining small values of parameters $D$ or $\Delta$: when $\Delta \leq 2$
or $D\leq 3$. Section  \ref{s4} contains open problems.

\section{A lower bound on the length of labeling schemes}\label{sec:feas}\label{s1}

As mentioned in the Introduction, topology recognition without any labels cannot be performed in any tree because no information can be successfully transmitted
in an unlabeled radio network. Hence, the length of a labeling scheme permitting topology recognition must be a positive integer.
In this section we show a lower bound
$\Omega(\log\log \Delta)$ on the length of labeling schemes that permit topology recognition for all trees with maximum degree $\Delta \geq 3$.

It is enough to consider trees with maximum degree $\Delta \ge 2^{36}$. Let $S$ be a star with the central node $r$ of degree $\Delta$. Denote one of the leaves of $S$ by $a$.
For $\lfloor \frac{\Delta}{2}\rfloor\le i \le \Delta-1 $, we construct a tree $T_i$ by attaching $i$ leaves to $a$. The maximum degree of each tree $T_i$ is $\Delta$. Let $\cT$ be the set of trees $T_i$, for  $\lfloor \frac{\Delta}{2}\rfloor\le i \le \Delta-1 $, cf. Fig. \ref{fig:fig1}. Hence the size of $\cT$ is at least $\frac{\Delta}{2}$.

Let $R$ be the set of leaves attached to $r$ and let $A$ be the set of leaves attached to $a$. For a tree $T\in \cT$, consider a labeling scheme $L(T)$ for the nodes of $T$. Let $R'\subseteq R$ and $A' \subseteq A$ be the sets of nodes with unique labels in $R$ and $A$, respectively. We define the {\em view} $\cV$ of the root $r$ as the labeled subtree induced by the nodes $r$, $a$ and by the sets of nodes $R'$ and $A'$.

\begin{figure}[h]
\centering
\includegraphics[width=0.5\textwidth]{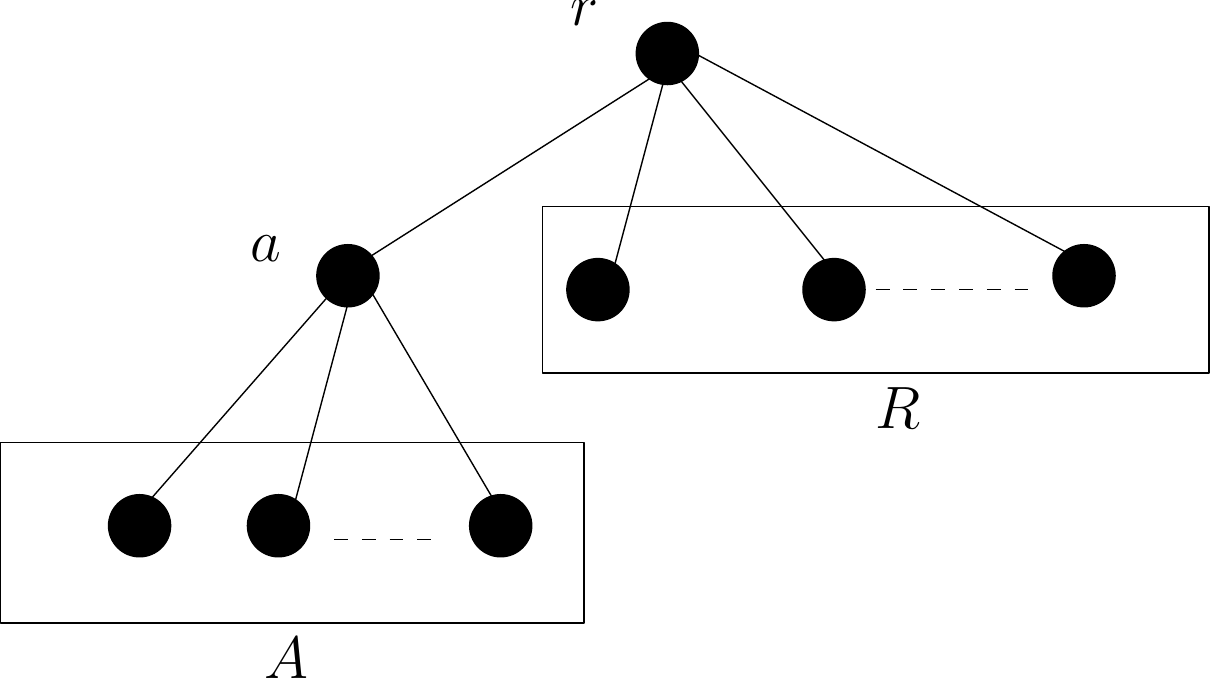}
\caption{Example of a tree in $\cal T$ }
\label{fig:fig1}
\end{figure}

We show that $\cV$ is the maximum information that the root $r$ can get at any time, when nodes of a tree $T\in \cT$ execute any deterministic algorithm.
Consider two nodes $v_1$ and $v_2$ in $R$ (respectively in $A$) with the same labels. Note that, the nodes in $R$ (respectively in $A$) can only hear from node $r$ (respectively from node $a$) and hence, the nodes $v_1$ and $v_2$ get the same information at any time. Since the labels of $v_1$ and $v_2$ are the same, therefore they must behave identically at any time, for any deterministic algorithm. Hence, the nodes $v_1$ and $v_2$ always broadcast in the same rounds,  and thus the node $r$ (respectively the node $a$) never receives any message from $v_1$ and $v_2$. The only nodes from which  $r$ (respectively $a$) can hear are the nodes in $R$ (respectively in $A$) with the unique labels.

The following result shows that any labeling scheme allowing topology recognition in trees of maximum degree $\Delta$ must have length $\Omega(\log \log \Delta)$.

\begin{theorem}\label{lem:feasibility}
For any tree $T \in \cT$ consider a labeling scheme LABEL($T$). Let {\it TOPO} be any topology recognition algorithm that solves topology recognition for every tree $T \in \cT$ using the scheme LABEL($T$).
Then there exists a tree $T' \in \cT$, for which the length of the scheme LABEL($T'$) is $\Omega(\log \log \Delta)$.
\end{theorem}

\begin{proof}
We prove this theorem by contradiction. Suppose that there exists an algorithm {\it TOPO} that solves topology recognition for every tree in $\cT$ with labels of length at most $\frac{1}{2}\log \log \Delta$. There are at most $2\sqrt{\log \Delta}$ possible different labels of this length.
There are at most $2^{2\sqrt{\log \Delta}}$ different possible subsets $R' \subset R$ with unique labels and at most  $2^{2\sqrt{\log \Delta}}$ different possible subsets $A' \subset A$  with unique labels, when the
length of the labels is at most $\frac{1}{2}\log \log \Delta$. Since each of the nodes $r$ and $a$ can also have one of the at most $2\sqrt{\log \Delta}$ possible labels, the total number of different views $\cV$ of $r$ is at most $(2 \sqrt{\log \Delta} ~2^{2\sqrt{\log \Delta}})^2 < \frac{\Delta}{2} \leq |\cT|$.

Therefore, by the Pigeonhole principle, there exist two trees $T'$, $T''$ in $\cT$ such that the view of $r$ in $T'$ is the same as the view of $r$ in $T''$. This implies that the node $r$ in $T'$ and the node $r$ in $T''$ must behave identically in every round, hence they must  output the same tree. This contradicts the fact the trees $T'$ and $T''$ are non-isomorphic.
This completes the proof.
\end{proof}

\section{Time for maximum degree $\Delta \geq 3$ and diameter $D \geq 4$  }\label{s2}

%
%

In this section, we present our main results concerning the time of topology recognition, using the shortest possible  labeling schemes (those of length $\Theta(\log\log \Delta)$) for trees of maximum degree
$\Delta \geq 3$ and diameter $D\geq 4$. We propose an algorithm using a labeling scheme of length $\Theta(\log\log \Delta)$ and working in time $O(D\Delta)$, and prove an almost matching lower bound
$\Omega (D\Delta^{\epsilon})$ on the time of such schemes, for any constant $\epsilon <1$.

\subsection{The main algorithm}

Let $T$ be a rooted tree of diameter $D$ and maximum degree $\Delta$.
It has either a central node or a central edge, depending on whether $D$ is even or odd.
 If $D$ is even, then the central node is the unique node in the middle of every simple
path of length $D$, and if $D$ is odd, then the central edge is the unique edge in
the middle of every simple path of length $D$. For the sake of description, we choose the central node or one of the endpoints of the central edge as the root $r$ of $T$.
Let $h=\lceil D/2 \rceil$ be the height of this tree. The {\em level} of any node $v$ is its distance from the root. For any node $v$ we denote by $T_v$ the subtree of $T$
rooted at $v$.

We propose an algorithm that solves topology recognition in time $O(D\Delta )$, using a labeling scheme of length $O(\log \log \Delta)$.
The structure of the tree will be transmitted bottom up, so that the root learns the topology of the tree, and then transmits it to all other nodes.
The main difficulty is to let every node know the round number $\rho$ in which it has to transmit, so that it is the only node among its siblings that transmits in round $\rho$,
and consequently its parent gets the message. Due to very short labels, $\rho$ cannot be explicitly given to the node as a part of its label. We overcome this difficulty by carefully coding $\rho$ for a node $v$, using the labels given to the nodes of the subtree rooted at $v$, so that $v$ can unambiguously decode $\rho$.

A node $v$ in $T$ is called {\em heavy}, if  $|V(T_v)|\ge \frac{1}{4} (\lfloor\log \Delta\rfloor+1)$. Otherwise, the node is called {\em light}. Note that the root is a heavy node. For a heavy node $v$, choose a subtree $T'_v$ of $T_v$ rooted at $v$, of size $\lceil\frac{1}{4} (\lfloor\log \Delta\rfloor+1)\rceil$.

First, we define the labeling scheme $\Lambda$. The label $\Lambda(v)$ of each node $v$ contains two parts. The first part is a vector of markers that are binary strings of constant length, used to identify nodes with different properties. The second part is a vector of 5 binary strings of length $O(\log \log \Delta)$ that are used to determine the time when the node should transmit.

Below we describe how the markers are assigned to different nodes of $T$.
\begin{enumerate}
\item Mark the root $r$ by the marker 0, and mark one of the leaves at maximum depth by the marker~1.
\item Mark all the nodes in $T'_r$ by the marker $2$.
\item Mark every heavy node by the marker $3$, and mark every light node by the marker $4$.
\item For every heavy node $v$ all of whose children are light, mark all the nodes of $T'_v$ by the marker~5.
\item For every light node $v$ whose parent is heavy, mark all the nodes in $T_v$ by the marker 6.
\end{enumerate}

The first part of every label is a binary string  $M$ of length 7, where the markers are stored. Note that a node can be marked by multiple markers. If the node is marked by the marker $i$, for $i=0,\dots ,6$, we have $M(i)=1$; otherwise, $M(i)=0$.

In order to describe the second part of each label, we define an integer $t_v$ for every heavy node $v \neq r$, and an integer $z_v$, for every light node $v$ whose parent is heavy.
 We define $t_v$, for a heavy node $v$ at level $l>0$,  to identify the time slot in which $v$ will transmit according to the algorithm.
 The definition is by induction on $l$. For $l=1$,
 let $v_1$, $v_2$, $\dots$, $v_x$, be the heavy children of $r$. Set $t_{v_i}=i$.
 Suppose that $t_v$ is defined for every heavy node $v$ at level $l$.
 Let $v$ be a heavy node at level $l$. Let $u_1$, $u_2$, $\dots$, $u_y$ be the heavy children of $v$.
We set $t_{u_1}=t_v$, and we define $t_{u_j}$, for  $2\le j\le y$, as distinct integers from the range $\{1,\dots, y\}\setminus \{t_v\}$.
This completes the definition of $t_v$, for all heavy nodes $v\neq r$.

We now define $z_v$, for a light node $v$ whose parent is heavy, to identify the time slot in which $v$ will transmit according to the algorithm.
Let $S_i$ be a maximal sequence of non-isomorphic rooted trees of $i$ nodes.  There are at most $2^{2(i-1)}$ such trees. Let $\cal S$ be the sequence which is the concatenation of $S_1$, $S_2$, $\dots$, $S_{\lceil\frac{1}{4} (\lfloor\log \Delta\rfloor+1)\rceil-1}$. Let $q$ be the length of $\cal S$. Then $q\le 2^{2 ( \frac{1}{4} (\lfloor\log \Delta\rfloor+1))}\le \sqrt{2\Delta}$. Note that the position of any tree of $i$ nodes in $\cal S$ is at most $2^{2i-1}$.
 Let ${\cal S}=(T_1,T_2,\dots,T_q)$. For a light node $v$ whose parent is heavy, we define $z_v=k$, if $T_v$ and $T_k$ are isomorphic.

The second part of each label is a vector $L$ of length 5, whose terms $L(i)$ are binary strings of length $O(\log \log \Delta)$. Initialize all terms $L(i)$ for every node $v$ to 0. We now describe how some of these terms are changed for some nodes.
They are defined as follows.

\begin{enumerate}
\item All the nodes which get $M(2)=1$ are the nodes of $T'_r$. There are exactly $\lceil \frac{1}{4} (\lfloor\log \Delta\rfloor+1)\rceil$ nodes in $T'_r$.
      All nodes in $T'_r$ are assigned distinct ids which are binary representations of the integers 1 to $\lceil \frac{1}{4} (\lfloor\log \Delta\rfloor+1)\rceil$. Let $s$ be the string of length $(\lfloor \log \Delta \rfloor +1)$ which is the binary representation of the integer $\Delta$. Let $b_1$, $b_2$, $\cdots$, $b_{\lceil \frac{1}{4} (\lfloor\log \Delta\rfloor+1)\rceil}$ be the substrings of $s$, each of length at most 4, such that $s$ is the concatenation of the substrings $b_1$, $b_2$, $\cdots$, $b_{\lceil \frac{1}{4} (\lfloor\log \Delta\rfloor+1)\rceil}$. The term $L(0)$
      corresponding to a node whose id is $i$, is set to the pair $(B(i),b_i)$, where $B(i)$ is the binary representation of the integer $i$.
      The intuitive role of the term $L(0)$ is to code the integer $\Delta$ in the nodes of the tree $T'_r$.

\item Let $v$ be a node with $M(3)=1$, and $M(5)=1$, i.e, let $v$ be a heavy node whose all children are light. All nodes in $T'_v$ are assigned distinct ids which are binary representations of integers 1 to $\lceil \frac{1}{4} (\lfloor\log \Delta\rfloor+1)\rceil$. Let $s$ be the string of length $(\lfloor \log \Delta \rfloor +1)$ which is the binary representation of the integer $t_v$. Let $b_1$, $b_2$, $\cdots$, $b_{\lceil\frac{1}{4}(\lfloor \log \Delta \rfloor +1)\rceil}$ be the substrings of $s$, each of length at most 4, such that $s$ is the concatenation of the substrings $b_1$, $b_2$, $\cdots$, $b_{\lceil\frac{1}{4}(\lfloor \log \Delta \rfloor +1)\rceil}$.
    The term $L(1)$ corresponding to a node whose id is $i$, is set to the pair $(B(i),b_i)$, where $B(i)$ is the binary representation of the integer $i$.
     The intuitive role of the term $L(1)$ is to code the integer $t_v$, for a heavy node $v$ whose all children are light, in the nodes of the tree $T'_v$.

    \item
    Let $v$ be a node with $M(3)=1$, i.e., a heavy node. Let $u$ be the parent of $v$. If  $t_u=t_v$, set $L(2)=1$ for the node $v$.
    The intuitive role of the term $L(2)$ at a heavy node $v$ is to tell its parent $u$ what is the value of $t_u$.

\item Let $v$ be a node with $M(4)=1$ and $M(6)=1$, i.e, let $v$ be a light node whose parent is heavy. All nodes in $T_v$ are assigned distinct ids which are binary representations of the integers 1 to $p$, where $p$ is the size of $T_v$.
    Let $s$ be the string of length at most $2p$ which is the binary representation of the integer $z_v$. Let $b_1$, $b_2$ $\cdots$, $b_{p}$ be the substrings of $s$, each of length at most 2, such that $s$ is the concatenation of the substrings $b_1$, $b_2$ $\cdots$, $b_{p}$.
    The term $L(3)$ of the node whose id is $i$ is set to the pair $(B(i),b_i)$, where $B(i)$ is the binary representation of the integer~$i$.
         The intuitive role of the term $L(3)$ is to code the integer $z_v$, for a light node $v$ whose parent is heavy, in the nodes of the tree $T_v$.

\item Let $v$ be a node with $M(3)=1$, i.e., a heavy node. Partition all light children $u$ of $v$ into sets with the same value of $z_u$. Consider any set
$\{u_1,u_2,\dots,u_a\}$  in this partition. Let $s$ be the binary representation of the integer $a$  and let $b_1$, $b_2$, $\cdots$, $b_{\lceil\frac{1}{4}(\lfloor \log a \rfloor +1)\rceil}$ be the substrings of $s$, each of length at most 4, such that $s$ is the concatenation of the substrings $b_1$, $b_2$, $\cdots$, $b_{\lceil\frac{1}{4}(\lfloor \log a \rfloor +1)\rceil}$.

For node $u_i$, where  $i \leq \lceil\frac{1}{4}(\lfloor \log a \rfloor +1)\rceil$,
the term $L(4)$ is set to the pair $(B(i),b_i)$, where $B(i)$ is the binary representation of the integer $i$, for $1\le i\le \lfloor \log a \rfloor +1$, and $b_i$ is the $i$th bit
of the binary representation of $a$. The intuitive role of the term $L(4)$ is to force two light children $v_1$ and $v_2$ of the same heavy parent, such that $z_{v_1}=z_{v_2}$, to transmit in different rounds.

\item
For any node $v$ the term $L(5)$ is set to the binary representation of  the integer $\lceil \frac{1}{4} (\lfloor\log \Delta\rfloor+1)\rceil$. This term will be used in a gossiping algorithm that is used as a subroutine in our algorithm.

\end{enumerate}

Notice that  the length of each $L(j)$ defined above is of length $O(\log \log \Delta)$ for every node, and there is no ambiguity in setting these terms, as every term for a node is modified at most once. This completes the description of our labeling scheme whose length is  $O(\log \log \Delta)$.

The algorithm consists of four procedures, namely Procedure {\tt Parameter Learning}, Procedure {\tt Slot Learning}, Procedure {\tt T-R} and Procedure {\tt Final}.
 In the first two procedures we will use the simple gossiping algorithm {\tt Round-Robin} which enables nodes of any  graph of size at most $m$ with distinct ids from the set  $\{1,\dots, m\}$ to gossip in time $m^2$,
 assuming that they know $m$ and that each node with id $i$ has an initial message $\mu_i$. The time segment $1,\dots , m^2$ is partitioned into $m$ segments of length $m$, and the node with id $i$ transmits in the $i$th round of each segment. In the first time segment, each node with id $i$ transmits the message $(i,\mu_i)$. 
 In the remaining $m-1$ time segments, nodes transmit all the previously acquired information.
 Thus at the end of  algorithm {\tt Round-Robin}, all nodes know the entire topology of the network, with nodes labeled by pairs $(i,\mu_i)$.

\noindent
 {\bf Procedure}  {\tt Parameter Learning}\\
  The aim of this procedure is for every node of the tree to learn the maximum degree $\Delta$, the level of the tree to which the node belongs, and the height $h$ of the tree.

 The procedure consists of two stages. The first stage is executed in rounds $1,\dots ,m^2$, where  $m=\lceil\frac{1}{4}(\lfloor \log \Delta \rfloor +1)\rceil$, and consists of performing algorithm {\tt Round-Robin} by the nodes with $M(2)=1$, i.e., the nodes in $T'_r$. Each such node uses its id $i$ written in the first component of the term $L(0)$, uses its label as $\mu_i$, and  takes $m$ as the integer whose representation is given in the term $L(5)$, and uses .

 After this stage, the node with $M(0)=1$, i.e., the root $r$, learns all pairs $(B(1),b_1)$, ..., $(B(m), b_m)$, where $B(i)$ is the binary representation of the integer $i$, corresponding to the term $L(0)$
 at the respective nodes. It computes the concatenation $s$ of the strings $b_1$, $b_2$, $\dots$, $b_{m}$. This is the binary representation of $\Delta$.

 The second stage of the procedure starts in round $m^2+1$. In round $m^2+1$, the root $r$ transmits the message $\mu$ that contains the value of $\Delta$. A node $v$, which receives the message $\mu$ at time $m^2+i$ for the first time, sets its level as $i$ and transmits $\mu$. When the node $u$ with $M(1)=1$, i.e., a deepest leaf, receives $\mu$ in round $m^2+j$, it sets its level as $h=j$, learns that the height of the tree is $h$, and transmits the pair $(h,h)$ in the next round. Every node at level $l$, after receiving the message $(h,l+1)$ (from a node of level $l+1$) learns $h$ and  transmits the pair $(h,l)$. After receiving the message $(h,1)$, the root $r$ transmits the message $\mu'$ that contains the value $h$.
  Every node learns $h$ after receiving it for the first time and retransmits $\mu'$, if its level is less than $h$. The stage, and hence the entire procedure, ends in round $m^2+3h$.

\noindent
 {\bf Procedure} {\tt Slot Learning}\\
 The aim of this procedure is for every heavy node all of whose children are light, and for every light node whose parent is heavy, to learn the time slot in which it should transmit. Moreover, at the end of the procedure, every light node $v$ learns $T_v$.

 Let $t_0=m^2+3h$, where  $m=\lceil\frac{1}{4}(\lfloor \log \Delta \rfloor +1)\rceil$. The total number of rounds reserved for this procedure is $2m^2$. The procedure starts in round $t_0+1$ and ends in round $t_0+2m^2$. The procedure consists of two stages. The first stage is executed in rounds $t_0+1,\dots, t_0+m^2$, and consists of performing algorithm {\tt Round-Robin} by the nodes with $L(1)\ne 0$, i.e., the nodes in $T'_v$, for a heavy node $v$ all of whose children are light.
 Each such node uses its id $i$ written in the first component of the term $L(1)$, uses its label as $\mu_i$, and
takes $m$ as the integer whose representation is given in the term $L(5)$.
 After this stage, each node $v$ with $M(3)=1$ and $M(5)=1$, i.e., a heavy node all of whose children are light, learns all pairs $(B(1),b_1),\dots ,(B(m), b_m)$, where $B(i)$ is the binary representation of the integer $i$, corresponding to the term $L(1)$
 at the respective nodes. It computes the concatenation $s$ of the strings $b_1$, $b_2$, $\dots$, $b_{m}$. This is the binary representation of the integer $t_v$, which will be used to compute the time slot in which node $v$ will transmit in the next procedure.

The second stage is executed in rounds $t_0+m^2+1,\dots, t_0+2m^2$, and consists of performing algorithm {\tt Round-Robin} by the nodes with $L(2)\ne 0$, i.e., the nodes in $T_v$, for a light node $v$ whose parent is heavy.
Each such node uses its id $i$ written in the first component of the term $L(3)$, uses its label as $\mu_i$, and
takes $m$ as the integer whose representation is given in the term $L(5)$.
After this stage, each node $v$ with $M(4)=1$ and $M(6)=1$, i.e., a light node whose parent is heavy, learns all pairs $(B(1),b_1),\dots ,(B(k), b_k)$, where $k< m$ and  $B(i)$ is the binary representation of the integer $i$, corresponding to the term $L(3)$
 at the respective nodes. Node $v$ computes the concatenation $s$ of the strings $b_1$, $b_2$, $\dots$, $b_{k}$. This is the binary representation of the integer $z_v$,  which will be used to compute the time slot in which node $v$ will transmit in the next procedure.
 Moreover, each node $w$ in $T_v$ learns $T_w$ because it knows the entire tree $T_v$ with all id's.
 The stage, and hence the entire procedure, ends in round $t_1=t_0+2m^2$.

 \noindent
 {\bf Procedure} {\tt T-R}\\
 The aim of this procedure is learning the topology of the tree by the root.

 All heavy nodes and all light nodes whose parent is heavy transmit in this procedure.
 The procedure is executed in $h$ epochs. The number of rounds reserved for an epoch is $2\Delta$. The first $\Delta$ rounds of an epoch are reserved for transmissions
 of heavy nodes and the last $\Delta$ rounds of an epoch are reserved for transmissions of light nodes whose parent is heavy.  The epoch $j$ starts in round $t_1+2(j-1)\Delta+1$ and ends in round $t_1+2j\Delta$. All the nodes at level $h-i+1$ which are either heavy nodes or  light nodes with a heavy parent transmit in the epoch $i$.
 When a node $v$ transmits in some epoch, it transmits a message $(\Lambda(v),T_v,C)$, where  $C=t_v$, if $v$ is a heavy node, and $C=0$, if it is a light node. Below we describe the steps that a node performs in the execution of the procedure, depending on its label.

 Let $v$ be a node with $M(4)=1$ and $M(6)=1$, i.e., $v$ is a light node whose parent is heavy. The node $v$ transmit in this procedure if $L(4) \ne 0$.
  Let the level of $v$ (learned in the execution of Procedure  {\tt Parameter Learning}) be $l$. Let the first component of the term $L(4)$ be the binary representation of the integer $c>0$. The node $v$ already knows the value $z_v$ which it learned in the execution of Procedure  {\tt Slot Learning}. Knowing $\Delta$, node $v$ computes the list ${\cal S}=(T_1,T_2,\dots,T_q)$ of trees (defined above) which unambiguously depends on $\Delta$.
The node $v$ transmits the message $(\Lambda(v),T_{z_v}, 0)$ in round $t_1+2(h-l)\Delta+\Delta+(z_v-1)\lceil\frac{1}{4}(\lfloor \log \Delta \rfloor+1)\rceil+c$. We will show that  node $v$ is the only node among its siblings that transmits in this round.

 Let $v$ be a node with $M(3)=1$ and $M(5)=1$, i.e., $v$ is a heavy node all of whose children are light. Let $l$ be the level of $v$. All the children of $v$ are light nodes with a heavy parent. They are at level $l-1$. Let $u_1$, $u_2$, $\dots$, $u_k$ be those children from which $v$ received messages in the previous epoch. First, the node $v$ partitions the nodes $u_1$, $u_2$, $\dots$, $u_{k}$ into disjoint sets $R_1,R_2,\cdots, R_e$ such that all nodes in the same set have sent the message with same tree $Q$.
 For each such set $R_d$, $1\le d\le e$,
 let $Q_d$ be the tree sent by nodes from $R_d$. The node $v$
 got all pairs $(B(1),b_1),\dots ,(B(x), b_x)$, where $x=|R_d|< m$ and  $B(i)$ is the binary representation of the integer $i$, corresponding to the term $L(4)$ at its children in $R_d$. Node $v$ computes the concatenation $s$ of the strings $b_1$, $b_2$, $\dots$, $b_{k}$. Let $y_d$ be the integer whose binary representation is $s$. After computing all $y_d$'s, for $1\le d\le e$, $v$ computes the tree $T_v$, by attaching $y_d$ copies of the the tree $Q_d$ to $v$ for  $d=1, \dots ,e$.  The node $v$ transmits the message $(\Lambda(v), T_v,t_v )$ in round $t_1+2(h-l)\Delta+t_v$. We will show that node $v$ is the only node among its siblings that transmits in this round.

Let $v$ be a node with $M(3)=1$ and $M(5)=0$, i.e., $v$ is a heavy node who has at least one heavy child. Let $u_1,\dots, u_{k_1}$ be the light children of $v$ from which $v$ received a  message in the previous epoch, and let $u'_1,\dots ,u'_{k_2}$ be the heavy children of $v$ from which $v$ received a message in the previous epoch. The node $v$ computes the tree $T_v$ rooted at $v$ as follows. It first attaches trees rooted at its light children,
using the messages it received from them, in the same way as explained in the previous case. Then, it attaches trees rooted at its heavy children.
These trees are computed from the code
$\beta$ in the message from  each of the heavy children of $v$.  Let $u'$ be the unique heavy child of $v$ for which the term $L(5)=1$. The node $v$ computes $t_v$ which is equal to the term $C$ in the message it received from the node $u'$.
The node $v$ transmits the message $(\Lambda(v),T_v,t_v )$ in round $t_1+2(h-l)\Delta+t_v$. We will show that node $v$ is the only node among its siblings that transmits in this round.

 \noindent
 {\bf Procedure} {\tt Final}\\
The aim of this procedure is for every node of the tree to learn the topology of the tree and to place itself in the tree.
The procedure starts in round $t_1+2h\Delta+1$ and ends in round $t_1+2h\Delta+h$. In round $t_1+2h\Delta+1$, the root $r$ transmits the message that contains the tree $T_r$. In general, every node $v$ transmits a message exactly once in Procedure {\tt Final}. This message contains the sequence
$(T_r, T_{w_p}, \dots ,T_{w_1},T_v)$, where $w_i$ is the ancestor of $v$ at distance $i$.
In view of the fact that every node $v$ already knows $T_v$ at this point,
after receiving a message containing the sequence $(T_r, T_{w_p}, \dots ,T_{w_1})$ in round $j$, a node $v$ transmits the sequence $(T_r, T_{w_p}, \dots ,T_{w_1},T_v)$ in round $j+1$,  if its level is less than $h$.

  A node $v$ outputs the tree $T_r$, and identifies itself as one of the nodes in $T_r$ for which the subtrees rooted at their ancestors in each level starting from the root are isomorphic to the trees
 in the sequence $(T_r, T_{w_p}, \dots ,T_{w_1},T_v)$. (Notice that there may be many such nodes). The procedure ends in round $t_1+2h\Delta+h$, when all nodes place themselves in $T_r$ and output $T_r$.

 Our algorithm can be succinctly formulated as follows.

 \begin{algorithm}
\caption{{\tt  Tree Topology Recognition}}
\begin{algorithmic}[1]
\STATE{\tt Parameter Learning}
 \STATE{\tt Slot Learning}
 \STATE{\tt T-R}
 \STATE{\tt Final}
\end{algorithmic}
\end{algorithm}

%
%

 We now prove the correctness of the algorithm and its time complexity. We will use the following lemmas.

\begin{lemma}
After the execution of Procedure  {\tt Parameter Learning}, every node learns $\Delta$, learns its level in the tree, and the height $h$ of the tree.
\end{lemma}

\begin{proof}
After round $m^2$, the root learns all pairs $(B(1),b_1)$, ..., $(B(m), b_m)$, where $B(i)$ is the binary representation of the integer $i$, corresponding to the term $L(0)$ at the respective nodes of the tree $T'_r$.
According to the assignment of the labels to the nodes, the binary string $s$, which is the concatenation of $b_1$, $b_2$, $\dots,$ $b_m$ is the binary representation of the integer $\Delta$. Therefore, in round $m^2$, the root $r$ learns $\Delta$.

In the second stage of Procedure  {\tt Parameter Learning}, $r$ transmits the message $\mu$ containing the value of $\Delta$ in round $m^2+1$. A node which is at distance $i$ from $r$ receives $\mu$ in round $m^2+i+1$ for the first time and learns its level $i$.
According to the labeling scheme $\Lambda$, one of the deepest leaves is marked by the marker 1, i.e, $M(1)=1$ for this node. When this node receives $\mu$, it learns its level $h$ and transmits a message that contains the value of $h$. When this message reaches $r$, $r$ learns $h$ and then $r$ transmits again the value of $h$. After receiving this message, every node learns the height $h$ of the tree.
\end{proof}
\begin{lemma}\label{lem:slot}
In the execution of Procedure  {\tt T-R}, if $v$ is a node that transmits in some round, then no other sibling of $v$ transmits in this round.
\end{lemma}
\begin{proof}
The nodes that transmit in Procedure  {\tt T-R} are either heavy or light with a heavy parent.

 According to the definition of $t_v$, the value $t_v$ for a heavy node $v$ is unique among the heavy siblings of $v$. We show that the node $v$ computes $t_v$ correctly in Procedure  {\tt T-R}.

 First assume that $v$ is a heavy node all of whose children are light. According to the labeling scheme $\Lambda$, the integer $t_v$ is coded using the terms $L(1)$ at the nodes in $T'_v$. The node $v$ correctly computes $t_v$ in Procedure {\tt Slot Learning} by collecting the terms $L(1)$ from all the nodes in $T'_v$.

Next assume that $v$ is a heavy node which has at least one heavy child. According to the labeling scheme $\Lambda$, there exists exactly one heavy child $u$ of $v$ such that $L(5)=1$ for $u$. This implies that $t_v=t_u$. In Procedure  {\tt T-R}, after receiving the messages from its children, the node $v$ learns $t_u$ and then sets $t_v=t_u$.  Therefore, $v$ computes $t_v$ correctly.

Let $v$ be a heavy node at level $i$. Therefore, it transmits in round $t_1+2(h-i)\Delta+t_v$. Since $t_v$ is unique and $v$ correctly computes $t_v$ before transmitting, $v$ is the only node among its siblings that transmits in this round.

A light node transmits in Procedure  {\tt T-R} only if the term $L(4) \ne 0$ at this node. Suppose that there exist two siblings $v_1$ and $v_2$ at level $i$ that transmit in the $(h-i+1)$-th epoch. Let the first components of the term $L(4)$ at $v_1$ and $v_2$ be, respectively, the integers $c_1\le\lceil\frac{1}{4}(\lfloor \log \Delta \rfloor+1)\rceil$ and $c_2\le\lceil\frac{1}{4}(\lfloor \log \Delta \rfloor+1)\rceil$. According to Procedure  {\tt T-R}, $v_1$ transmits in round $\tau_1=t_1+2(h-i)\Delta+\Delta+(z_{v_1}-1)\lceil\frac{1}{4}(\lfloor \log \Delta \rfloor+1)\rceil+c_1$ and $v_2$ transmits in round $\tau_2= t_1+2(h-l)\Delta+\Delta+(z_{v_2}-1)\lceil\frac{1}{4}(\lfloor \log \Delta \rfloor+1)\rceil+c_2$. If $z_{v_1} \ne z_{v_2}$, then $\tau_1\ne \tau_2$. Hence suppose that $z_{v_1} = z_{v_2}$. According to the labeling scheme $\Lambda$, if
 $v_1$ and $v_2$ are siblings and satisfy $z_{v_1}=z_{v_2}$, then the first components of the term $L(4)$ at these nodes are distinct integers  $c_1 \ne c_2$.
Therefore, $\tau_1\ne \tau_2$. This proves that no two light siblings transmit in the same round in Procedure  {\tt T-R}.

It remains to consider the case of siblings $v_1$ and $v_2$, such that $v_1$ is heavy and $v_2$ is light. According to Procedure  {\tt T-R}, $v_1$ transmits in the time interval $[t_1 +2(h-l)\Delta +1, t_1 +2(h-l)\Delta + \Delta]$, and $v_2$ transmits after time $t_1 +2(h-l)\Delta + \Delta$. Hence they do not transmit in the same round.
\end{proof}

\begin{lemma}\label{lem:topology}
After the $(h-j)$-th epoch of Procedure  {\tt T-R}, every heavy node $v$ at level $j$ correctly computes $T_v$.
\end{lemma}
\begin{proof}
We prove this lemma by induction on the level $j$. For the base case, consider a heavy node without heavy children. We prove that such a node $v$  correctly computes $T_v$.

Let $v$ be a heavy node  at level $j$ all of whose children are light. Let $u_1$, $u_2$, $\dots$, $u_k$ be those children from which $v$ received messages in the epoch
$h-j$. According to the Procedure  {\tt T-R}, let $R_1,R_2,\cdots, R_e$ be the disjoint sets of children of $v$ such that all nodes in the same set have sent the message with the same tree. For $1\le d\le e$, let  $R_d=\{u'_1, u_2', \dots,u'_a\}$ and let the corresponding tree be $Q_d$. All the nodes in $R_d$ are light and the trees rooted at these nodes are $Q_d$. The node $v$ got the pairs $(B(1),b_1)$, ..., $(B(a), b_a)$, where $B(i)$ is the binary representation of the integer $i$, corresponding to the term $L(4)$ at the respective nodes.
According to the labeling scheme $\Lambda$, the string $s$ which is the concatenation of the strings $b_1, \dots, b_a$ is the binary representation of the integer $g_d$, where $g_d$ is the total number of children of $v$ which have the same tree  $Q_d$ rooted at them. This implies that in the tree $T_v$, there are $g_d$ copies of $Q_d$ attached to $v$, for $1\le d\le e$. This proves that $v$ correctly computes $T_v$ after receiving all the messages from its children in epoch $h-j$. Thus the base case is proved.

As the induction hypothesis, suppose that every heavy node of level $j-1$ correctly computes $T_v$.
Let $v$ be a heavy node at level $j$. If $v$ has no heavy child, then the lemma is true according to the base case. Let $v$ be a node with at least one heavy child.
Let $u_1,\dots, u_{k_1}$ be the light children of $v$ from which $v$ received a  message in epoch $h-j$, and let $u'_1,\dots ,u'_{k_2}$ be the heavy children of $v$ from which $v$ received a message in epoch $h-j$. Now, the tree $T_v$ is computed by attaching to $v$ subtrees rooted at its light children and subtrees rooted at its heavy children. The first attachments are explained in the base case. By Lemma \ref{lem:slot}, all heavy children of  $v$ transmit in different rounds. Hence, $v$ receives messages from each of its heavy children, and the message includes the subtree rooted at the respective heavy child. By the induction hypothesis, each of these children of $v$ computed the subtree rooted at itself correctly, and sent it to $v$ in the $(h-j)$-th epoch.
 The node $v$ attaches to itself all these subtrees rooted at its heavy children. Therefore, $v$ computes $T_v$ correctly.
 This proves the lemma by induction.
\end{proof}

We are now ready to prove our  main positive result.

\begin{theorem}
Upon completion of Algorithm {\tt  Tree Topology Recognition}, all nodes of a tree correctly output the topology of the tree and place themselves in it. The algorithm uses labels  of length
$O(\log\log \Delta)$ and works in time $O(D\Delta)$, for trees of maximum degree $\Delta$ and diameter $D$.
\end{theorem}

\begin{proof}
Since $r$ is a heavy node at level zero, according to Lemma \ref{lem:topology}, $r$ computes the tree $T_r$ at the end of the $h$-th epoch of Procedure {\tt T-R}.
$T_r$ is the entire tree.
 After computing it, $r$ transmits this tree to all nodes during Procedure {\tt Final}. Hence upon completion of this procedure, every node learns the topology of the tree $T$. Since, in Procedure {\tt Final}, every node learns additionally the sequence of subtrees rooted at all of its ancestors, it places itself correctly in $T_r$. This proves the correctness of Algorithm {\tt  Tree Topology Recognition}.

According to the labeling scheme $\Lambda$, the label of every node has two parts. The first part is a vector $M$ of constant length and each term of $M$ is of constant length. The second part is a vector $L$ of constant length and each term of $L$ is of length $O(\log \log \Delta)$. Therefore, the length of the labeling scheme $\Lambda$ is $O(\log\log \Delta)$.

Algorithm {\tt  Tree Topology Recognition} ends after round $t_1+2h\Delta=t_0+2m^2+2h\Delta=3m^2+3h+h\Delta$.
Since $m$ is in $O(\log \Delta)$ and $h$ is in $O(D)$,
the time complexity of Algorithm {\tt  Tree Topology Recognition} is $O(D\Delta)$.
\end{proof}

\subsection{The lower bound}

In this section, we prove that any topology recognition algorithm using a labeling scheme of length $O(\log\log \Delta)$
 must use time at least $\Omega(D\Delta^{\epsilon})$, for any constant $\epsilon<1$,  on some tree of diameter $D \geq 4$ and maximum degree $\Delta \geq 3$.
 We split the proof of this lower bound into three parts, corresponding to different ranges of the above parameters, as the proof is different in each case.
 
 \vspace*{0.5cm}
\noindent
{\bf Case 1:  $\Delta$ bounded,  $D$ unbounded}

In this case we need to show a lower bound $\Omega(D)$.

\begin{lemma}\label{time2}
Let $D  \geq 4$ be any integer,  let $\Delta \ge 3$ be any integer constant and let $c>1$ be any real constant.
For any tree $T$ of maximum degree $\Delta$ consider a labeling scheme LABEL($T$) of length at most $c\log \log \Delta$. Let $TOPO$ be any algorithm that solves topology recognition for every tree $T$ of maximum degree $\Delta$ using the labeling scheme LABEL($T$). Then there exists a tree $T$
of maximum degree $\Delta$ and diameter $D$ for which $TOPO $ must take time $\Omega(D)$.
\end{lemma}
\begin{proof}
We first assume that $D$ is even. The case when $D$ is odd will be explained later.
It is enough to prove the lemma for $D>7$. We construct a class of trees $\cT(D)$ as follows.
Let $h_1=\lfloor\frac{D+8c-6}{8c}\rfloor>1$. Let $h_2=\frac{D}{2}$. Since $c>1$, we have $h_2>h_1$. Also, $h_2\ge 4c(h_1-1)-3$.

Let $A_1$ be a line of length $h_1-1$ with endpoints $r$ and $s_1$. Let $A_2$ be a rooted tree of height $h_2-1$ with root $s_2$  of degree $\Delta-1$, such that every non-leaf node other than $s_2$  has degree $\Delta$.
Let $T$ be the tree rooted at $r$, constructed by adding the edge between the nodes $s_1$ and $s_2$, and attaching an additional node of degree 1 to each of the two leaves of $A_2$ which are at distance $2(h_2-1)$.
The total number of leaves of $T$  at level $h_1+h_2-1$ is $(\Delta-1)^{h_2-1}-2$. Let $v_1,v_2,\cdots,v_{(\Delta-1)^{h_2-1}-2}$ be the leaves of $T$ at level $h_1+h_2-1$. Let $(x_1,x_2,\cdots,x_{(\Delta-1)^{h_2-1}-2})$ be a sequence of integers where $0\le x_i\le \Delta-1$. We construct a tree $T_x$ from $T$ by attaching $x_i$ leaves to the node $v_i$, for $1 \le i\le (\Delta-1)^{h_2-1}-2$. The diameter of the tree $T_x$ is $2h_2=D$.
Let $\cT(D)$ be a maximal set of pairwise non-isomorphic trees among the trees $T_x$, cf. Fig. \ref{fig:fig7}. 

\begin{figure}[h]
\centering
\includegraphics[width=0.5\textwidth]{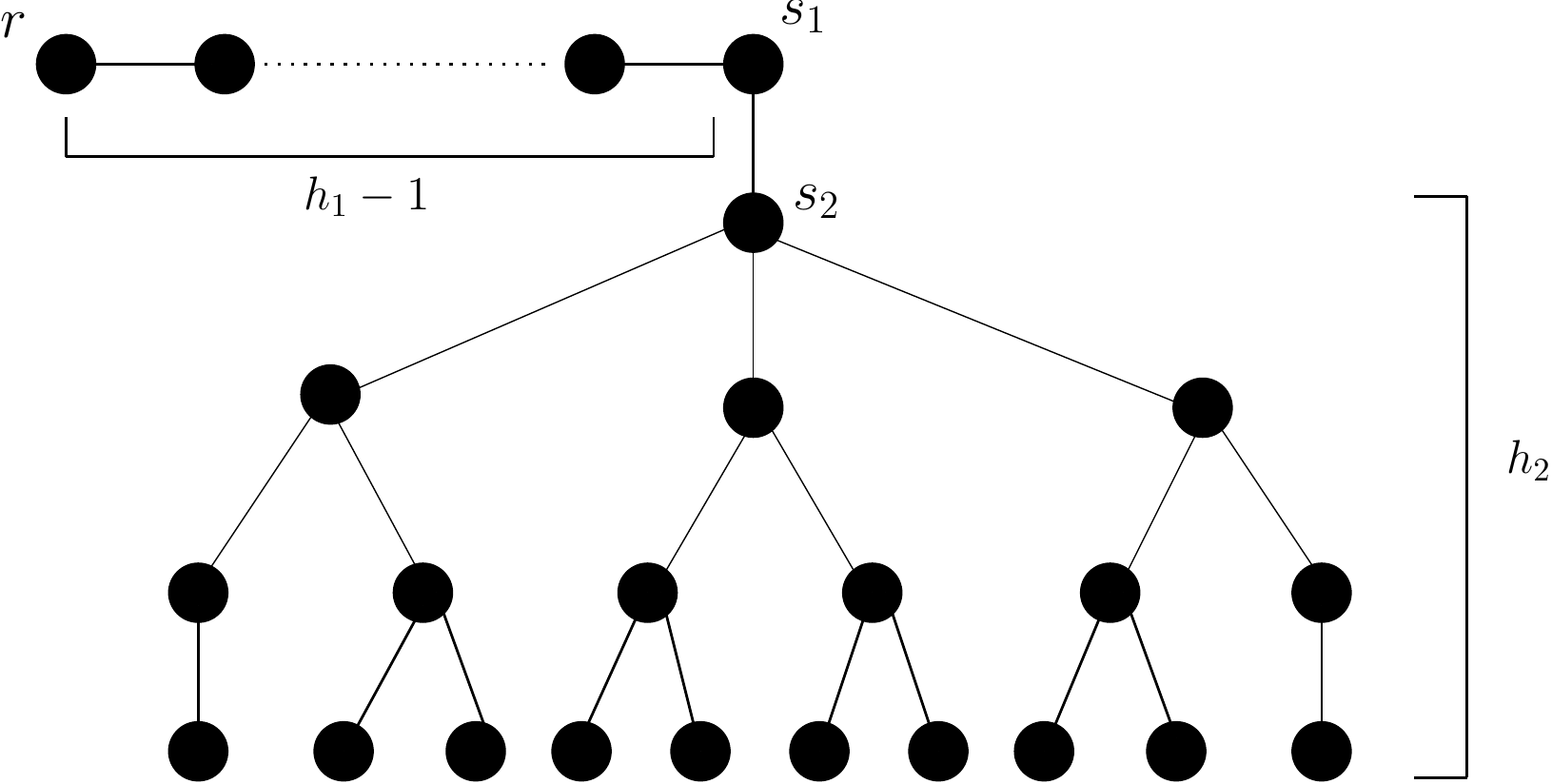}
\caption{Example of a tree in $\cT(D)$ for $D$ even}
\label{fig:fig7}
\end{figure}

Let $k=(\Delta-1)^{h_2-1}-2$ and let $z=(z_1,z_2,\cdots,z_{\Delta})$ be the sequence of integers, such that $0 \le z_i\le k$ and $\sum_{i=1}^{k}z_i=k$. The number of  such sequences $z$ is ${k+\Delta-1 \choose \Delta-1}$.

 Let $T(z)$ be a tree  in $\cT(D)$, such that for each $1\le j\le \Delta$, the number of nodes in $\{v_1,v_2,\cdots,v_{k}\}$ of degree $j$ in $T(z)$ is $z_j$.  Note that for two sequences $z'\ne z''$, the trees $T(z')$ and $T(z'')$ are non-isomorphic. Let $\cT'$ be a maximal set of pairwise non-isomorphic trees among the trees $T(z)$. Then, $|\cT'|= {k+\Delta-1 \choose \Delta-1}$. Since $\cT' \subset \cT(D)$, therefore, $|\cT(D)|\geq{k+\Delta-1 \choose \Delta-1} \ge (\frac{k+\Delta-1}{\Delta-1})^{\Delta-1}$.

 Consider an algorithm {\it TOPO} that solves topology recognition for every tree $T'$ in $\cT(D)$ in time at most $\tau=h_1-2$. The history $H(TOPO, \tau)$ of the root $r$ is the labeled subtree of $T'$ spanned by the nodes in the line $A_1$.

 There are at most $2^{c\log \log \Delta +1}=2(\log \Delta)^c$ different possible labels.
Hence, the total number of possible histories $H(TOPO, \tau)$ is at most $\left(2(\log \Delta)^c\right)^{h_1-1}$. Let $X=\left(2(\log \Delta)^c\right)^{h_1-1}=2^{h_1-1+c(h_1-1)\log\log \Delta}$.
For sufficiently large $D$, we have\\ $|\cT(D)|\ge  (\frac{k+\Delta-1}{\Delta-1})^{\Delta-1}=(\frac{(\Delta-1)^{h_2-1}-2+\Delta-1}{\Delta-1})^{\Delta-1} \ge ((\Delta-1)^{h_2-2}-2)^{\Delta -1}>2^{\frac{h_2-2}{2}(\Delta-1)\log (\Delta-1)}$.
Hence $|\cT(D)|> 2^{\frac{h_2-2}{2}(\Delta-1)\log (\Delta-1)}\ge 2^{\frac{4c(h_1-1)-5}{2}\log (\Delta-1)} > 2^{h_1-1+c(h_1-1)\log\log \Delta}=X$, for sufficiently large $D$.

Therefore, for sufficiently large $D$, there exist two trees $T'$ and $T''$ in $\cT(D)$, such that the roots of the two trees have the same history. It follows that the root $r$ in $T'$ and the node $r$ in $T''$ output the same tree as the topology.
This is a contradiction, which proves the lemma for even $D$.
For odd $D$, do the above construction for $D-1$ and attach one additional  node of degree 1 to one of the leaves.
Then the same proof works with $D$ replaced by $D-1$.
\end{proof}

\vspace*{0.5cm}
\noindent
{\bf Case 2:  $\Delta$ unbounded,  $D$ bounded}

In this case, we need to show a lower bound $\Omega(\Delta^{\epsilon})$, for any constant $\epsilon<1$. The following lemma proves a stronger result.

\begin{lemma}\label{time3}
Let $\Delta \ge 3$ be any integer, let $D \ge 4$ be any integer constant, and let $c>0$ be any real constant.
For any tree $T$ of maximum degree $\Delta$, consider a labeling scheme LABEL($T$) of length at most $c\log \log \Delta$. Let $TOPO$ be an algorithm that solves topology recognition for every tree of maximum degree $\Delta$ and diameter $D$ using the labeling scheme LABEL($T$). Then there exists a tree $T$ of maximum degree $\Delta$ and diameter $D$ for which $TOPO $ must take time $\Omega(\frac{\Delta}{(\log \Delta)^c})$.
\end{lemma}

\begin{proof}
Let $S$ be a star with the root $s$, where the degree of $s$ is $\Delta$. Let $v_1,v_2,\cdots,v_{\Delta}$ be the leaves. Let $S'$ be a line of length $D-3$ with endpoints $r$ and $s'$. Let $x=(x_1,x_2,\cdots,x_{\Delta})$ be a sequence of integers such that $\lfloor \frac{\Delta}{2}\rfloor\le x_i \le {\Delta-1}$, for $1\le i\le \Delta$. We construct a tree $T_x$  rooted at the node $r$ by adding the edge  between the nodes $s$ and $s'$ and attaching $x_i$ leaves to the node $v_i$, for $1\le i\le \Delta$.
Let $\cT(\Delta)$ be a maximal set of pairwise non-isomorphic trees among the trees $T_x$, cf. Fig. \ref{fig:fig6}.
Then $|\cT(\Delta)|\ge {\Delta+\lceil\frac{\Delta}{2}\rceil-1 \choose \lceil\frac{\Delta}{2}\rceil -1}\ge \left(\frac{{\frac{3\Delta}{2}-1}}{\frac{\Delta}{2}}\right)^{\frac{\Delta}{2}-1}\ge 2^{\frac{\Delta}{2}-1}$.

\begin{figure}[h]
\centering
\includegraphics[width=0.5\textwidth]{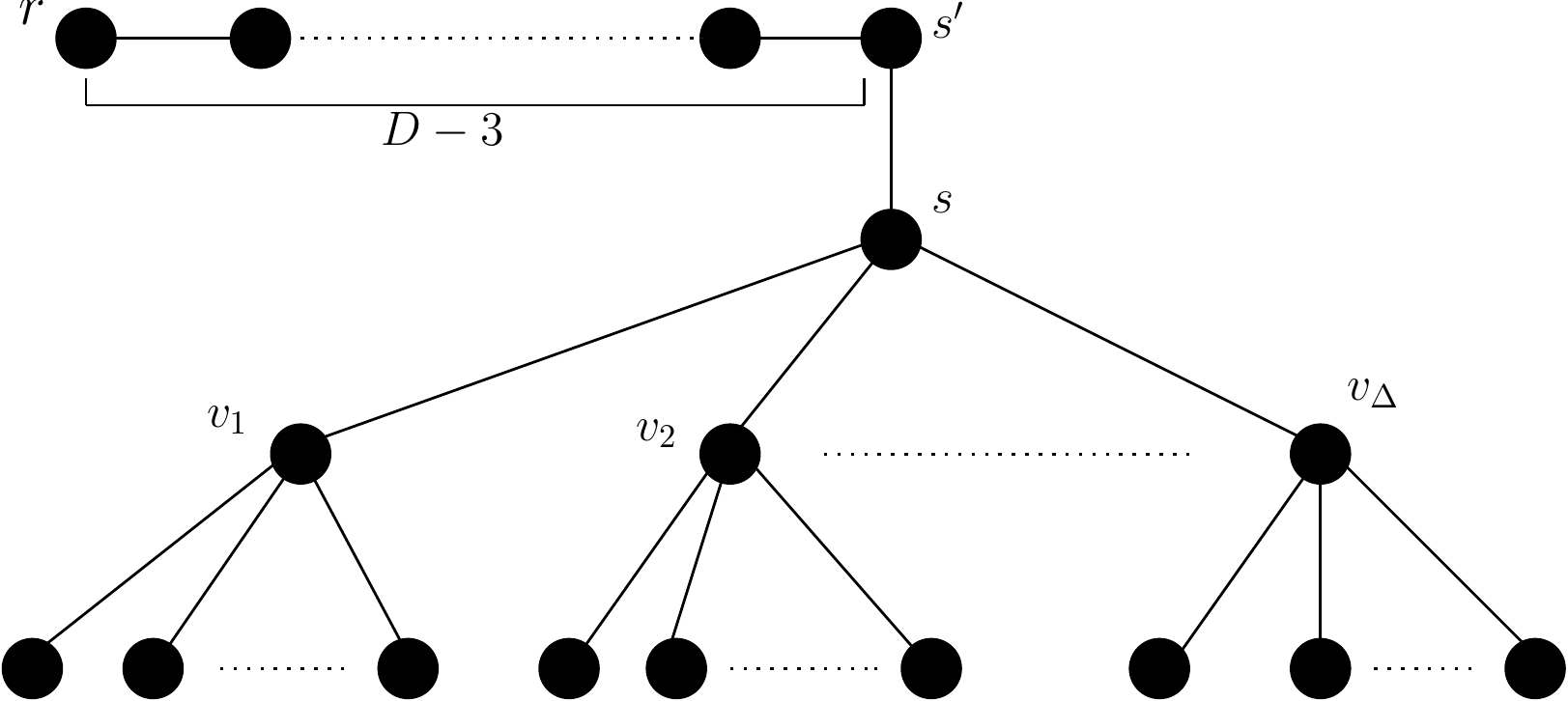}
\caption{Example of a tree in $\cT(\Delta)$ for $D$ even}
\label{fig:fig6}
\end{figure}

 Consider an algorithm {\it TOPO} that solves topology recognition for every tree $T$ in $\cT(\Delta)$ in time at most $\tau\le\frac{\Delta}{8(\log \Delta)^c}$. Now, as explained in Section \ref{sec:feas}, the only children from which the node $v_i$ can hear are the children of $v_i$ with the unique labels. The set of labels of the children from which the node $v_i$ can hear is a subset of the set of all possible labels of length at most $c\log\log \Delta$.

 There are at most $2{(\log \Delta)}^c$ possible different labels of this length,
and there are at most $2^{2{(\log \Delta)}^c}$ different possible subsets of the set of possible labels. Since, in time $\tau$, at most $\tau$ nodes from the set $\{v_1,v_2,\cdots,v_{\Delta}\}$ can successfully transmit to $r$, the history $H(TOPO,\tau)$ is a labeled subtree of the tree $T$ that  contains the  path from $r$ to $s$,
at most $\tau$ nodes from the set  $\{v_1,v_2,\cdots,v_{\Delta}\}$, and all the children with unique labels of each of these at most $\tau$ nodes.
Therefore,
the number of possible histories $H(TOPO,\tau)$ of the root $r$ is at most $(2{(\log \Delta)}^c)^{D-1}\left(2{(\log \Delta)}^c 2^{2{(\log \Delta)}^c}\right)^{\tau}$. Let $X=(2{(\log \Delta)}^c)^{D-1}\left(2{(\log \Delta)}^c 2^{2{(\log \Delta)}^c}\right)^{\tau}$. Then $$X=(2{(\log \Delta)}^c)^{D-1+\tau}\left( 2^{2\tau{(\log \Delta)}^c}\right)\le (2{(\log \Delta)}^c)^{D-1+\tau} \cdot  2^{\frac{\Delta}{4}}<2^{\frac{\Delta}{2}-1}\le |\cT(\Delta)|,$$ for sufficiently large $\Delta$. Therefore, for sufficiently large $\Delta$, there exist two trees $T'$ and $T''$ in $\cT(\Delta)$, such that the roots of the two trees have the same history. If follows that the root $r$ in $T'$ and the root $r$ in $T''$ output the same tree as the topology. This is a contradiction, which proves the lemma.
\end{proof}

\vspace*{0.5cm}
\noindent
{\bf Case 3: unbounded $\Delta$ and $D$ }

Let $\Delta \ge 3$, $D \ge 4$ be integers. We first assume that $D$ is even. The case when $D$ is odd will be explained later. It is enough to prove the lower bound for
$D \ge 6$.
Let $h=\lfloor\frac{D}{6}\rfloor$ and $g=\frac{D}{2}-h$. Then $2h \le g \le 2h+2$.
Let $P$ be a line of length $g$ with nodes $v_1$, $v_2$, $\cdots$, $v_{g+1}$, where $v_1$ and $v_{g+1}$ are the endpoints of $P$. We construct from $P$ a class of trees called {\it sticks} as follows.

Let $x=(x_1,x_2,\cdots,x_{g})$ be a sequence of integers, with $0 \le x_i \le \Delta-2$. Construct a tree $P_x$ by attaching $x_i$ leaves to the node $v_i$ for $1 \le i \le g$, cf. Fig. \ref{fig:fig3}.  Let $\cP$ be the set of all sticks constructed from $P$. Then $|\cP|=(\Delta-1)^{g} $. Let $\cP=\{P_1,P_2,\cdots, P_{(\Delta-1)^{g}}\}$.

Let $S$ be a rooted tree of height $h$, with root $r$ of degree $\Delta -1$,  and with all other non-leaf nodes of degree $\Delta$. The nodes in $S$ are called {\it basic nodes}.
Let $Z=\{w_1,w_2, \cdots ,w_z\}$, where $z= {(\Delta-1)^{h}}$, be the set of leaves of $S$. Consider a sequence $y=(y_1,y_2, \cdots, y_z)$, for $1\le y_i \le (\Delta-1)^{g}$. We construct a tree $T_y$ from $S$ by attaching to it the sticks in the following way: each leaf $w_i$ is identified with the node  $v_1$ of the stick $P_{y_i}$, for $1\le i\le z$, cf. Fig. \ref{fig:fig2}. We will say that the stick $P_{y_i}$ is {\em glued} to node $w_i$. The diameter of each tree $T_y$ is $D$.
For odd $D$, do the above construction for $D-1$ and attach one additional  node of degree 1 to one of the leaves.

\begin{figure}[h]
\centering
\includegraphics[width=0.5\textwidth]{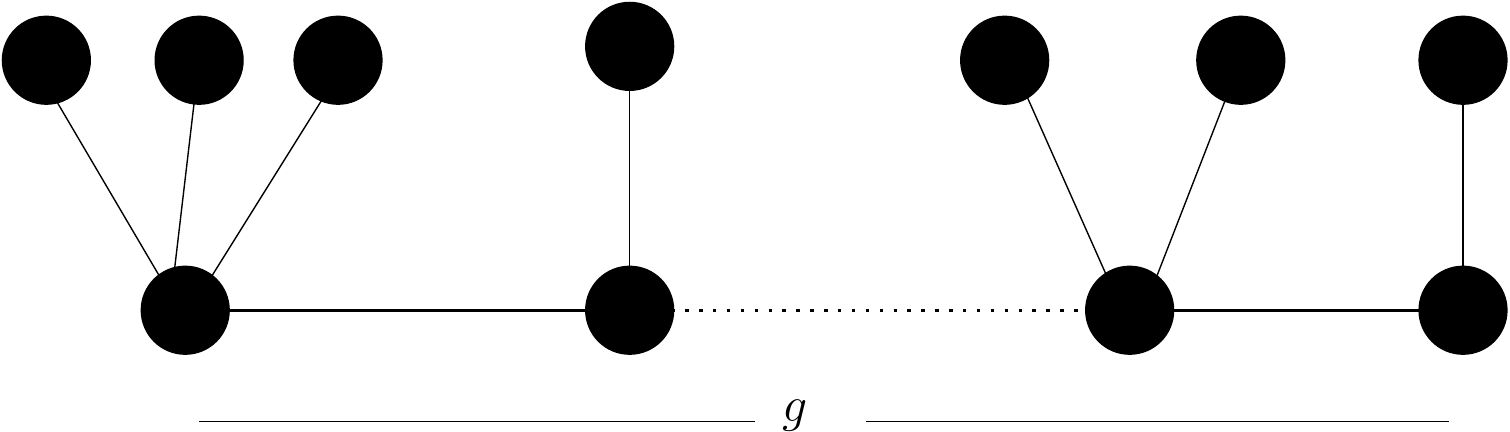}
\caption{Example of a stick}
\label{fig:fig3}
\end{figure}

\begin{figure}[h]
\centering
\includegraphics[width=0.5\textwidth]{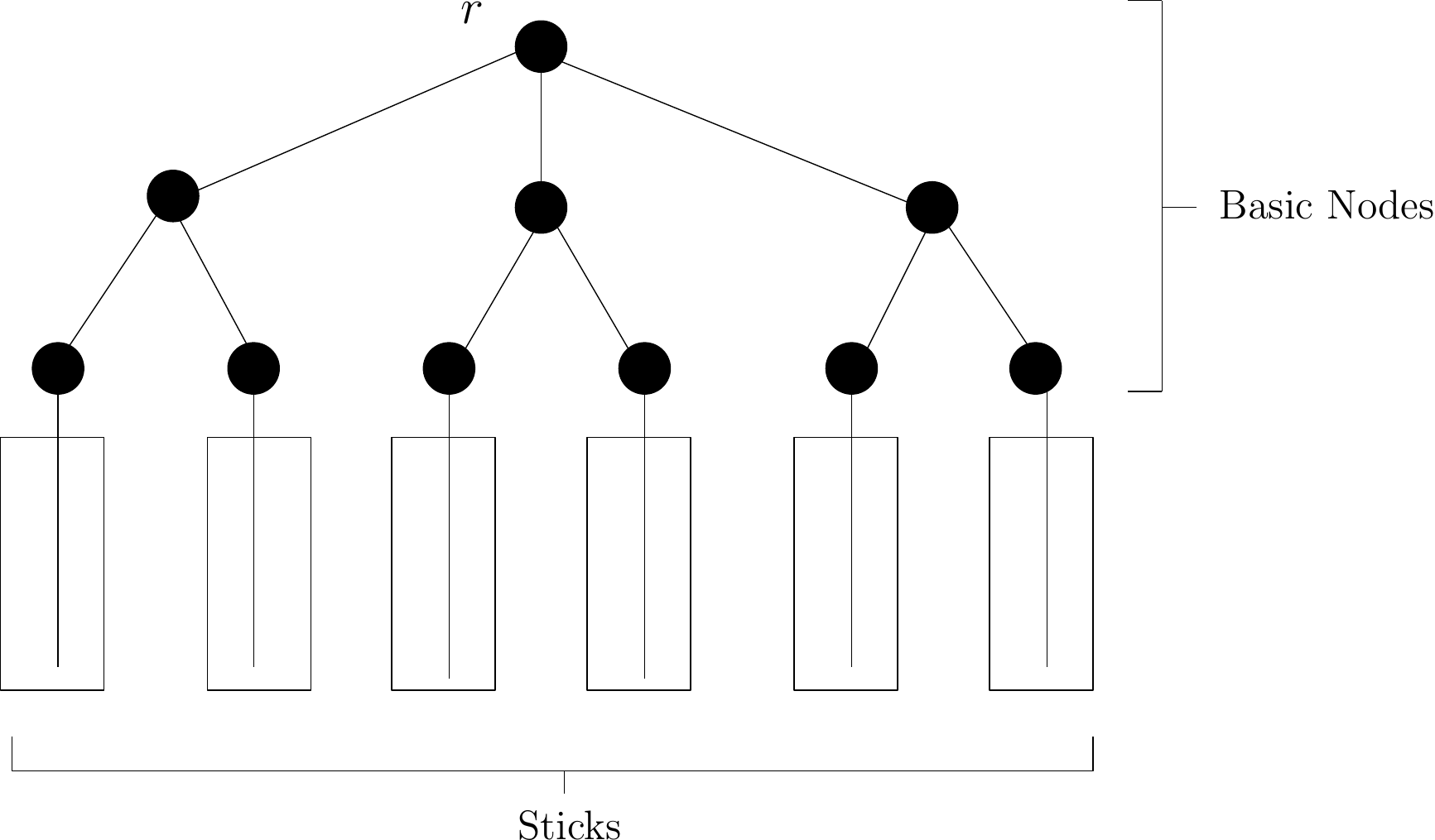}
\caption{An Example of a tree  in $\cT(\Delta,D)$ with basic nodes and sticks}
\label{fig:fig2}
\end{figure}

Let $\cT(\Delta,D)$ be a maximal set of pairwise non-isomorphic trees among the trees $T_y$. Then, $|\cT(\Delta,D)| \ge \frac{((\Delta-1)^{g})^z}{z!} \ge \frac{((\Delta-1)^{g})^z}{z!}   \ge (\Delta -1)^{h (\Delta-1)^h}$.
%

Consider any time $\tau >0$.
For any tree $T \in \cT(\Delta,D)$, consider any labeling scheme $L(T)$ and let $\cA$ be any algorithm that solves topology recognition in every tree $T \in  \cT(\Delta,D)$ in time $\tau$, using the labeling scheme $L(T)$. The following lemma gives an upper bound on the number of basic nodes that can belong to a history of the root $r$.

\begin{lemma}\label{lem:time}
Let $B$ be the number of basic nodes of level $i$ that can reach $r$ within time $\tau$, according to algorithm $\cA$. Then  $ B \le \frac{\tau^{i}}{i!}$ if $\tau\ge i$, and $B=0$, otherwise.
\end{lemma}
\begin{proof}
Any node $v$ at level $i$ cannot reach $r$ in time smaller than $i$. Therefore, $B=0$ for $\tau<i$.

For $\tau \ge i$, we prove the lemma by induction.  The lemma is true for  $i=1$, as the number of basic nodes of level 1 that reach $r$ within time $\tau$ is at most $\tau$. Suppose that the lemma is true for the nodes of levels $\le j$.

Let $u_1,u_2,\cdots, u_\ell$ be the children of $r$ that reach $r$ within time $\tau$. Let $t_1' > t_2'\cdots >t_\ell'$ be the last rounds in which the nodes $u_1,u_2,\cdots, u_\ell$, respectively, transmit. Then, $t_i' \le \tau-i+1$, for $1 \le i \le l$. Let $F(u_i,t')$ be the number of nodes at level $j+1$ of $T$ that reach node $u_i$ within time $t'$.
The nodes in the $(j+1)$-th level in $T$ are the nodes in the $j$-th level in the subtrees which are rooted at children of $r$. Hence, by the induction hypothesis, the number of nodes at level $j+1$ that reach $r$ within time $\tau$ is at most $\sum_{i=1}^{l} F(u_i,\tau-i+1) \le \sum_{i=1}^{l} \frac{(\tau-i+1)^j}{j!} = \frac{1}{j!}\sum_{i=1}^{l} (\tau-i+1)^j\le \frac{1}{j!}\sum_{i=1}^{\tau} i^j \le \frac{1}{j!}  \int_0^{\tau} x^j\,dx =\frac{\tau^{j+1}}{(j+1)!}$.
This implies that the statement of the lemma is true for the level $j+1$. Hence, the lemma follows by induction.
\end{proof}

The next lemma gives the announced lower bound on the time of topology recognition for the class $\cT(\Delta,D)$.

\begin{lemma}\label{time1}
Let $\epsilon<1$ be any positive real constant, and let $c>1$ be  any real constant.
For any tree $T \in \cT(\Delta,D)$, consider a labeling scheme LABEL($T$) of length at most $  c\log \log \Delta$. Then there exist integers $\Delta_0, D_0>0$ such that any  algorithm that solves topology recognition for every tree $T \in \cT(\Delta,D)$, where $\Delta \geq \Delta_0$ and $D \geq D_0$,  using the scheme LABEL($T$), must take time $\Omega(D\Delta^\epsilon)$ for some tree $T\in \cT(\Delta,D)$.
\end{lemma}

\begin{proof}
Wee first do the proof for even $D$.
Consider an algorithm {\it TOPO} that solves topology recognition for every tree $T\in \cT(\Delta,D)$ in time $\tau \le (\frac{D}{6}-1)\Delta^\epsilon \le h\Delta^\epsilon$ with a labeling scheme $LABEL(T)$ of length at most $ c\log \log \Delta$.
For a scheme of this length, there are at most $2^{c\log \log \Delta +1}=2(\log \Delta)^c$ different possible labels.
According to Lemma \ref{lem:time}, for $1 \le i \le h$ the number of basic nodes of level $i$, that reach $r$ within time $\tau$ is at most $\frac{\tau ^i}{i!}$, if $\tau\ge i$, otherwise there are no such nodes.

Denote by $q$ the total number of basic nodes that reach $r$ within time $\tau$.
If $\tau \ge h$, then $q\le \sum_{i=1}^{h} \frac{\tau ^i}{i!} \le h \frac{\tau ^{h}}{h} = h \frac{(h\Delta^\epsilon)^h}{h!}$.
We know that $\log (h!)=h \log h -\frac{h}{\ln 2}+\frac{1}{2} \log h +O(1) \ge h \log h -\frac{h}{\ln 2}$. Since $\ln 2 > \frac{1}{2}$, we have $\log (h!) > h \log h-2h$. Therefore, $h! > \frac{h^{h}}{2^{-2{h}}}$, and hence
 $q \le h\Delta^{h\epsilon}2^{2h}$.
If $\tau<h$, then $q\le \sum_{i=1}^{\tau} \frac{\tau ^i}{i!}\le \tau \Delta^{\tau\epsilon}2^{2\tau}\le h \Delta^{h\epsilon}2^{2h}$. Therefore, $q\le h \Delta^{h\epsilon}2^{2h}$, for all $\tau>0$.

The number of different unlabeled sticks is at most $(\Delta-1)^{2h+2}$. Nodes of each such stick can be labeled with labels of length at most $\lfloor c\log \log \Delta\rfloor$ in at most $\left(2(\log \Delta)^c\right)^{(2h+2)\Delta}$ ways, because each stick can have at most $(2h+2)\Delta$ nodes.
Therefore, the number of different labeled sticks is at most $p=(\Delta-1)^{2h+2} \left(2(\log \Delta)^c\right)^{(2h+2)\Delta}$. 

The history of the root $r$ of a tree $T\in \cT(\Delta,D)$ may include some nodes from a stick in $T$ only if the basic node at level $h$ to which this stick is glued is a node in the history. The maximum information that the root can get from a basic node $v$ at level $h$, but not from any other node at this level, is the information about the whole labeled stick glued to $v$.

The number of possible histories $H(TOPO, \tau)$ of the node $r$ is at most the product of  the number of possible labelings of the basic nodes in $H(TOPO, \tau)$ and the number of possible gluings of labeled sticks to them. Since there are at most  $q$ basic nodes in $H(TOPO, \tau)$, there are at most $(2 (\log \Delta)^c)^q$ possible labelings of these nodes. Since there are at most $p$ labeled sticks to choose from, the number of possible gluings of labeled sticks to the basic nodes in
$H(TOPO, \tau)$ is at most $p^q$. Therefore,
the number of possible histories $H(TOPO, \tau)$ of the node $r$ is at most $2^q(\log \Delta)^{cq} p^q =\left(2p(\log \Delta)^c\right)^q$.
Let $X=\left(2p(\log \Delta)^c\right)^q$.
 We have $\log X= q(\log p+ 1+c\log\log \Delta)= q+q \log p+ qc \log\log \Delta$.
  Also, $\log p=(2h+2)\log (\Delta-1)+(2h+2)\Delta(1+ c\log \log \Delta).$ Therefore, $\log X=q(1+\log p+c\log\log \Delta)= q\left(1+(2h+2)\log (\Delta-1)+(2h+2)\Delta(1+ c\log \log \Delta)+c\log\log \Delta\right) \le 5qc(2h+2)\Delta\log \Delta \le 5h \Delta^{h\epsilon+1}2^{2h} c(2h+2)\log \Delta$.
 Also, $\log |\cT(\Delta,D)|\geq  h(\Delta-1)^h \log (\Delta-1)$. Now, for any $\Delta$ and for sufficiently large $h$, we have $5h \Delta^{h\epsilon+1}2^{2h} c(2h+2)<\frac{1}{2}h\Delta^h$. Therefore, $5h \Delta^{h\epsilon+1}2^{2h} c(2h+2)\log \Delta < \frac{1}{2}h\Delta^h \log \Delta< h(\Delta-1)^h \log (\Delta-1)$, for sufficiently large $\Delta$
 and sufficiently large $h$.

 It follows that, for sufficiently large $h$ and $\Delta$, we have $\log X < \log |\cT(\Delta,D)|$.
Therefore, there exist integers $\Delta_0$ and $D_0$ such that $X <  |\cT(\Delta,D)|$, for all $\Delta\ge \Delta_0$ and $D\ge D_0$.
 Hence, for $\Delta \ge \Delta_0 $ and $D \ge D_0$, there exist two trees $T_1$ and $T_2$ in $\cT(\Delta,D)$ whose roots have the same history. Therefore, the root $r$ in $T_1$ and the root $r$ in $T_2$ output the same tree as the topology, within time $\tau$. This is a contradiction, which proves the lemma for even $D$. For odd $D$, the same proof works with $D$ replaced by $D-1$.
\end{proof}

Lemmas  \ref{time2} , \ref{time3} and  \ref{time1} imply the following theorem.

\begin{theorem}
Let $\epsilon<1$ be any positive real number.
For any tree $T$ of maximum degree $\Delta \geq 3$ and diameter $D \geq 4$, consider a labeling scheme of length $O(\log \log \Delta)$. Then any topology recognition algorithm using such a scheme for every tree $T$ must take time $\Omega(D\Delta^\epsilon)$ for some tree.
\end{theorem}

\section{Time for small maximum degree $\Delta$ or small diameter $D$}\label{s3}

In this section we solve our problem for the remaining cases of small parameters $\Delta$ and $D$, namely, in the case when $\Delta \leq 2$ or $D\leq 3$.
We start with the case of small diameter $D$.

\subsection{Diameter $D=3$}

First we propose a topology recognition algorithm for all trees of diameter $D=3$ and of maximum degree $\Delta \geq 3$, using a labeling scheme of length $O(\log \log \Delta)$
and  working in time $O(\frac{\log \Delta}{\log \log \Delta})$.

\vspace*{0.5cm}
\noindent
{\bf Algorithm} {\tt Small Diameter T-R}

Let $T$ be a tree of diameter $3$ and maximum degree $\Delta \geq 3$, rooted at node $r$.  The node $r$ is one of the endpoints of the central edge of $T$. Since $D=3$, $r$ has exactly one child $a$ of degree larger than one, and all other children of $r$ are leaves.
Below we describe the assignment of the labels to the nodes of $T$.

\begin{enumerate}
\item The root $r$ gets the label $0$.
\item Let $u_1,u_2,\cdots,u_{k_1}$ be the children of $r$ which are leaves. Let $s$ be the string of length $(\lfloor \log k_1\rfloor+1)$ which is the binary representation of the integer $k_1$. Let $p=\lceil\frac{\lfloor \log k_1\rfloor+1}{\lfloor\log \log \Delta\rfloor}\rceil$.
    Let $b_1$, $b_2$, $\cdots$, $b_{p}$ be the substrings of $s$, each of length at most $\lfloor\log \log \Delta\rfloor$, such that $s$ is the concatenation of the substrings $b_1$, $b_2$, $\cdots$, $b_{p}$. For $1 \le i\le p-1$, the node $u_i$ gets the label $(0,B(i),b_i)$, where
    $B(i)$ is the binary representation of the integer $i$. The node $u_{p}$ gets the label $(1,B(p),b_p)$, where $B(p)$ is the binary representation of $p$. For $i >p$, the node $u_i$ gets the label $(0,0,0)$.
    \item The node $a$ gets the label $f$, where $f$ is the binary representation of the integer $p$.
\item Let $u'_1,u'_2,\cdots,u'_{k_2}$ be the children of $a$. These are leaves. Let $s'$ be the string of length $(\lfloor \log k_2\rfloor+1)$ which is the binary representation of the integer $k_2$. Let $q=\lceil\frac{\lfloor \log k_2\rfloor+1}{\lfloor\log \log \Delta\rfloor}\rceil$.
    Let $b'_1$, $b'_2$, $\cdots$, $b'_{q}$ be the substrings of $s$, each of length at most $\lfloor\log \log \Delta\rfloor$, such that $s'$ is the concatenation of the substrings $b'_1$, $b'_2$, $\cdots$, $b'_{q}$. For $1 \le i\le q-1$, the node $u'_i$ gets the label $(0,B(i),b'_i)$, where
    $B(i)$ is the binary representation of the integer $i$. The node $u'_{q}$ gets the label $(1,B(q),b'_q)$, where $B(q)$ is the binary representation of $q$. For $i >q$, the node $u'_i$ gets the label $(0,0,0)$.
 \end{enumerate}

 Every node transmits according to its label. Let $v$ be any node whose label contains three components. If the second component of the label is the binary representation of an integer $c>0$, then $v$ transmits a message that contains its label, in round $c$.

 The node with label $f$ (i.e., node $a$) waits until it gets a message from a node with a 3-component label, whose first component is 1. The node $a$
 gets all pairs $(0,B(1),b_1)$, $(0,B(2),b_2)$,..., $(1,B(x), b_x)$, where  $B(i)$ is the binary representation of the integer $i$. The node $a$ computes the concatenation $s$ of the strings $b_1$, $b_2$, $\dots$, $b_{x}$. Let $y_1$ be the integer whose binary representation is $s$. The node $a$ computes $p$ from $f$ and transmits the message  $(p,y_1)$ in round $\max\{p+1,x+1\}$.

 Similarly, the node with label 0 (i.e., node $r$) waits until it gets a message from a node with a 3-component label whose first component is 1, and it computes the integer $y_2$ from the messages it got until then, as explained above for the node $a$ computing $y_1$. The node $r$ waits for the message that arrives next. This message is $(p,y_1)$. Node $r$ computes the tree $T'$ by attaching $y_2+1$ children to $r$ and attaching $y_1$ leaves to one of these children. Then $r$ transmits $T'$. When the node with label $f$ gets the message with the tree $T'$, it learns $T'$ and retransmits it. Every node outputs the tree $T'$ after getting the message. The nodes $a$ and $r$ identify themselves in the topology by looking at their own labels. A node which learned $T'$ from $a$ identifies itself as a child of $a$ and a node which learned $T'$ from the node $r$ identifies itself as a child of $r$.

 The following lemma estimates the performance of Algorithm {\tt Small Diameter T-R}.

\begin{lemma}\label{ub:D=3}
Algorithm {\tt Small Diameter T-R} solves topology recognition for trees of maximum degree $\Delta \geq 3$ and diameter $D=3$, in time $O(\frac{\log \Delta}{\log\log \Delta})$, using labels  of length
$O(\log\log \Delta)$.
\end{lemma}
\begin{proof}
By definition, the length of the labeling scheme is $O(\log\log \Delta)$. It remains to estimate the execution time of the algorithm.
Let $y_1$ be the number of leaves attached to $a$ and $y_2$ be the number of leaves attached to $r$. According to the label assignments to the nodes in $T$, the $\lceil\frac{\lfloor \log y_1\rfloor+1}{\lfloor\log \log \Delta\rfloor}\rceil$ leaves attached to $a$ get distinct labels, whose first components represents integers from 1 to $\lceil\frac{\lfloor \log y_1\rfloor+1}{\lfloor\log \log \Delta\rfloor}\rceil$ and the concatenation of the second components represents $y_1$. Therefore, the node $a$ computes $y_1$ correctly after round $\lceil\frac{\lfloor \log y_1\rfloor+1}{\lfloor\log \log \Delta\rfloor}\rceil$. Similarly, the node $r$ computes $y_2$ correctly after round $\lceil\frac{\lfloor \log y_2\rfloor+1}{\lfloor\log \log \Delta\rfloor}\rceil$. The node $a$ is the only node which transmits in round $\max\{\lceil\frac{\lfloor \log y_1\rfloor+1}{\lfloor\log \log \Delta\rfloor}\rceil+1, \lceil\frac{\lfloor \log y_2\rfloor+1}{\lfloor\log \log \Delta\rfloor}\rceil+1\}$ and it transmits a message whose second component is the value $y_1$. After getting this message from $a$, the node $r$ learns $y_1$ and $y_2$, and hence it computes the topology of the tree correctly. Every other node learns the topology within the next round after the node $r$ transmits the topology of the tree. The algorithm ends in round $\max\{\lceil\frac{\lfloor \log y_1\rfloor+1}{\lfloor\log \log \Delta\rfloor}\rceil+1, \lceil\frac{\lfloor \log y_2\rfloor+1}{\lfloor\log \log \Delta\rfloor}\rceil+1\}+1$. Since $y_1,y_2 \le \Delta$, the time complexity of the algorithm is $O(\frac{\log \Delta}{\log\log \Delta})$.
\end{proof}

The following lemma gives a lower bound on the time of topology recognition for trees of diameter 3, matching the performance of Algorithm {\tt Small Diameter T-R}.
\begin{lemma}\label{lb:D=3}
Let $\Delta \ge 3$ be any integer, and let $c>0$ be any real constant.
For any tree $T$ of maximum degree $\Delta$ consider a labeling scheme LABEL($T$) of length at most $c\log \log \Delta$. Let $TOPO$ be any algorithm that solves topology recognition for every tree of maximum degree $\Delta$ and diameter $3$ using the labeling scheme LABEL($T$). Then, for every $\Delta \ge 2$,  there exists a tree $T$ of maximum degree $\Delta$ and diameter 3, for which $TOPO $ must take time $\Omega(\frac{\log \Delta}{(\log \log \Delta)})$.
\end{lemma}

\begin{proof}
We use the class  $\cT$ of trees from Section \ref{sec:feas}.
Consider an algorithm {\it TOPO} that solves topology recognition for every tree $T\in \cT$ in time $\tau \le \frac{\log \Delta}{4c\log\log \Delta}$ using a labeling scheme $LABEL(T)$ of length at most $  c\log \log \Delta$.

In time $\tau$, at most $\tau$ nodes can reach the node $a$, and at most $\tau$ nodes can reach the node $r$. Since there are at most $2(\log \Delta)^c$ different possible labels of length at most $c\log \log \Delta$, the total number of possible histories $H(TOPO, \tau)$ of the root $r$ is at most\\ $(2(\log \Delta)^c)^{2\tau+2} \le 2^{2(\frac{\log \Delta}{4c\log\log \Delta}+1)} ({\log \Delta)}^{\frac{\log \Delta}{2\log\log \Delta}+2}< \frac{\Delta}{2} \leq |\cT|$, for sufficiently large $\Delta$.

Therefore, for sufficiently large $\Delta$, there exist two trees $T'$ and $T''$ in $\cT$ such that the roots of the two trees have the same history. Hence the root $r$ in $T'$ and the root $r$ in $T''$ output the same tree as the topology. This is a contradiction, which proves the lemma.
\end{proof}

In view of  Lemmas \ref{ub:D=3} and  \ref{lb:D=3}, we have the following result.

\begin{theorem}
The optimal time for topology recognition in the class of trees of  diameter $D=3$ and maximum degree $\Delta \geq 3$,
using a labeling scheme of length $\Theta(\log\log \Delta)$, is $\Theta(\frac{\log \Delta}{(\log \log \Delta)})$.
\end{theorem}

\subsection{Diameter $D=2$}

We now consider the case of trees of diameter 2, i.e., the class of stars. Since there is exactly one star of a given maximum degree $\Delta$, the problem of topology recognition for $D=2$ and a given maximum degree $\Delta$ is trivial. A meaningful variation of the problem for $D=2$ is to consider all trees (stars) of maximum degree {\em at most} $\Delta$, for a given $\Delta$.

Let $T$ be a star with the central node $r$. The labeling scheme and the algorithm for topology recognition in $T$ are similar to Algorithm {\tt Small Diameter T-R} in the case of $D=3$.
The objective of the algorithm is for every node to learn the value of $\Delta$. A set of $\lceil\frac{1+\lfloor \log \Delta\rfloor}{\log\log \Delta}\rceil$ leaves are given distinct labels. Each such label contains two components. The first components are distinct ids from
1 to $\lceil\frac{1+\lfloor \log \Delta\rfloor}{\log\log \Delta}\rceil$, and the second components are the substrings of length $\lfloor\log \log \Delta\rfloor$ whose concatenations in increasing order of ids is the binary representation of $\Delta$. Leaves with distinct labels transmit one by one in every round, in the order of their ids, and after $O(\frac{\log \Delta}{\log \log \Delta})$ rounds, the node $r$ computes the value of $\Delta$. Then  $r$ transmits $\Delta$. Every leaf and the node $r$ output a star with degree $\Delta$, and every node places itself in this star either as the root or as a leaf.

The following lemma gives a lower bound on the time of topology recognition for stars, matching the above upper bound.
\begin{lemma}\label{star}
Let $\Delta \ge 3$ be any integer, and let $c>0$ be any real constant.
For any star $T$ of maximum degree at most $\Delta$ consider a labeling scheme LABEL($T$) of length at most $c\log \log \Delta$. Let $TOPO$ be any algorithm that solves topology recognition for every star of maximum degree at most $\Delta$, using the labeling scheme LABEL($T$). Then  there exists a star $T$ of maximum degree at most $\Delta$, for which $TOPO $ must take time $\Omega(\frac{\log \Delta}{(\log \log \Delta)})$.
\end{lemma}
\begin{proof}
Let $T_j$,
for $j=1,2,\dots, \Delta-\lfloor\frac{\Delta}{2}\rfloor$, be the star with the central node $r$ and degree $\lfloor\frac{\Delta}{2}\rfloor+j$.
Let $\cT$ be the set of all trees $T_j$, $j=1,2,\dots, \Delta-\lfloor\frac{\Delta}{2}\rfloor$. Then $|\cT| \ge \frac{\Delta}{2}$.
Consider an algorithm {\it TOPO} that solves topology recognition for every star $T\in \cT$ in time $\tau \le \frac{\log \Delta}{2c(1+\log\log \Delta)}-1$ using a labeling scheme $LABEL(T)$ of length at most $  c\log \log \Delta$.

In time $\tau$, at most $\tau$ nodes can reach the node $r$. Since there are at most $2(\log \Delta)^c$ different possible
labels of length at most $c\log \log \Delta$, the total number of possible histories $H(TOPO, \tau)$ of the root $r$ is at most
$(2(\log \Delta)^c)^{\tau+1} \le  2^{(c+c\log\log \Delta) \cdot \frac{\log \Delta}{2c(1+\log\log \Delta)}} <\frac{\Delta}{2} <|\cT|$, for sufficiently large $\Delta$.

Therefore, for sufficiently large $\Delta$, there exist two trees $T'$ and $T''$ in $\cT$ such that the roots of the two trees have the same history. Hence the root $r$ in $T'$ and the root $r$ in $T''$ output the same tree as the topology. This is a contradiction, which proves the lemma.
\end{proof}

In view of the above described algorithm and of Lemma \ref{star}, we have the following result.

\begin{theorem}
The optimal time for topology recognition in the class of trees of diameter $D=2$ (i.e., stars) and maximum degree at most $\Delta$, where $\Delta \geq 3$,
using a labeling scheme of length $\Theta(\log\log \Delta)$, is $\Theta(\frac{\log \Delta}{(\log \log \Delta)})$.
\end{theorem}

 \subsection{Maximum degree  $\Delta=2$}

We finally address the case of trees of maximum degree $\Delta=2$, i.e., the class of lines. Since there is exactly one line of a given diameter $D$, the problem of topology recognition for $\Delta=2$ and for a given diameter $D$ is trivial. A meaningful variation of the problem for $\Delta=2$ is to consider all trees (lines) of diameter
{\em at most} $D$, for a given $D$.

We first propose a topology recognition algorithm for all lines of diameter at most $D$, where $D\ge 4$,  using a labeling scheme of length $O(1)$ and  working in time $O(\log D)$.

\vspace*{0.5cm}
\noindent
{\bf Algorithm} {\tt Line-Topology-Recognition}

 Let $T$ be a tree of maximum degree $2$ and diameter at most $D$, i.e.,  a line of length at most $D$.
Let $v_1, v_2, \dots,v_{k+1}$, for $k \leq D$, be the nodes of $T$, where $v_1$ and $v_{k+1}$ are the two endpoints.
At a high level, we partition the line into segments of length $O(\log k)$ and assign labels, containing (among other terms) couples of bits, to the nodes in each segment. This is done in such a way that the concatenation of the first bits of the couples in a segment is the binary representation of the integer $k$, and the concatenation of the second
bits  of the couples in a segment is the binary representation of the segment number.  In time $O(\log k)$, every node learns the labels in each segment, and computes $k$ and the number $j \geq 0$ of the segment to which it belongs. It identifies its position in this segment from the round number in which it receives a message
for the first time. Then a node outputs the line of length $k$ with its position in it.

Below we describe the assignment of the labels to the nodes of $T$. The label of a node $v$ is a quadruple $(\alpha_v, \beta_v, \gamma_v, \delta_v)$. The term $\alpha_v$ is the binary representation of an integer from the set $\{0,1,2,3\}$ which represents the type of the node $v$, to be specified later. The term $\beta_v$ is a bit of the binary representation of the integer $k$. The  term $\gamma_v$ is a bit of the binary representation of the number $j$ of the segment. The term $\delta_v$ is
the binary representation of an integer from the set $\{0,1,2\}$ which represents the distance of the node (mod 3)  from one of the endpoints. Hence each label has a constant length. More precisely, the labels are assigned as follows.
\begin{enumerate}
\item For $1 \leq j \leq \lfloor \frac{k}{3+\lfloor \log k \rfloor}\rfloor-1$, the node $v_{j(3+\lfloor \log k \rfloor)}$ gets the label $(0,0,0,e)$, where $e=j(3+\lfloor \log k \rfloor)\mod 3$. The node $v_{k+1}$ gets the label $(0,0,0,(k+1)\mod 3 )$. These nodes are called {\it type 0} nodes.

\item For  $0\le j\le \lfloor \frac{k}{3+\lfloor \log k \rfloor}\rfloor-2$, the node $v_{j(3+\lfloor \log k \rfloor)+1}$ gets the label $(1,0,0,e)$, where $e=(j(3+\lfloor \log k \rfloor)+1) \mod 3$. These nodes are called {\it type 1} nodes.

\item For $0\le j\le \lfloor \frac{k}{3+\lfloor \log k \rfloor}\rfloor-2$, $1\le i\le 3+\lfloor \log k \rfloor-1$, the node $v_{j(3+\lfloor \log k \rfloor)+i+1}$ gets the label $(10,b_i,b_i',e)$, where $b_i$, $b_i'$ are the $i$-th bit of the binary representation of $k$ and $j$, respectively,  and $e=(j(3+\lfloor \log k \rfloor)+i+1) \mod 3$. These nodes are called {\it type 2} nodes.

\item All other nodes $v_i$ gets the label $(11,0,0,e)$, where $e=i \mod 3$. These nodes are called {\it type 3} nodes.
\end{enumerate}

Fig. \ref{fig:fig5} shows the location of nodes of different types in the line.

\begin{figure}[h]
\centering
\includegraphics[width=0.5\textwidth]{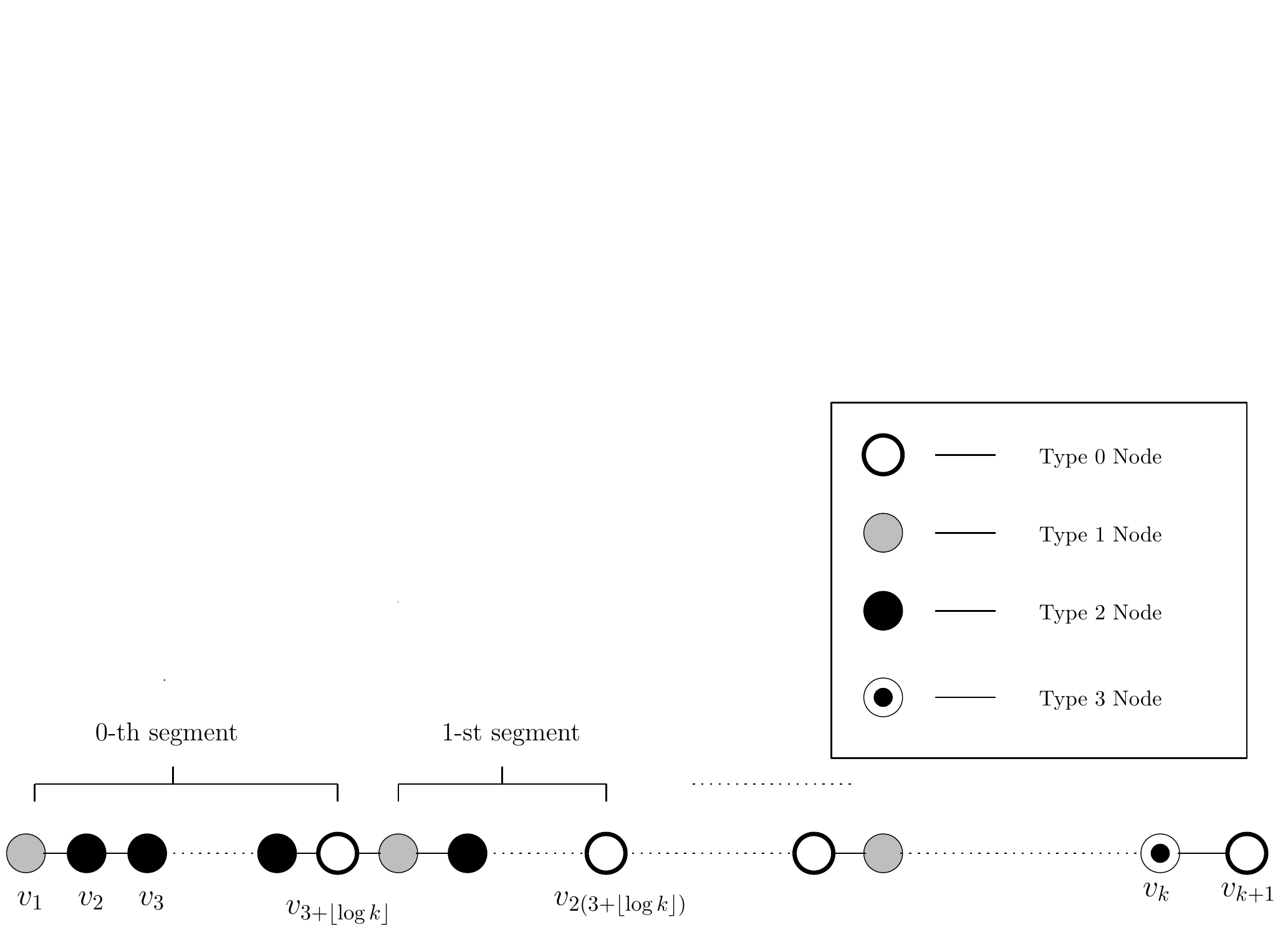}
\caption{The partition of a line into segments}
\label{fig:fig5}
\end{figure}

We now describe the algorithm for topology recognition in lines, using  the above labeling scheme. For a node $v$, if the last component of its label represents the integer $l \in \{0,1,2\}$, then a round is called {\it dedicated} to $v$, if its number is of the form $3i+l$, for  $i \ge 0$.
In the algorithm, nodes that transmit simultaneously have the same dedicated round, and hence the distance between them is at least 3. This prevents collisions.
A node $v$ can identify its type by looking at the first component of its label.

Every node $v$ keeps a variable $r_v$ which will be set to some round number in which $v$ gets a specific message.
 First, the type 1 nodes transmit the message $(0,\epsilon,\epsilon)$ in round 1, where $\epsilon$ is the empty string. After receiving a message $(0,s,s')$ for the first time in round $i$, a node $v$ of type 2 sets $r_v :=i$ and transmits the message $(0,s\cdot (\beta_v),s'\cdot (\gamma_v))$ in the next dedicated round, where `$\cdot$' denotes the concatenation operation on strings. After receiving a message $(0,s,s')$ for the first time in round $i$, a node $v$ of type 3 sets $r_v :=i$ and retransmits the message $(0,s,s')$ in the next dedicated round. After receiving a message $(0,s,s')$, where $s,s'\ne \epsilon$, for the first time in round $i$, a node $v$ of type 0 sets $r_v :=i$, and  computes the integers $k$ and $j$, whose binary representations are $s$ and $s'$, respectively. Then it outputs a line of length $k$, and identifies itself as the node $v_{jk+r_v+1}$. Finally, the node transmits the message $(k,j,\delta_v)$ in the next dedicated round.  After receiving a message $(k,j,e)$ for the first time, a node $v$ of type 2 or of type 3 learns $k$ and $j$. Then it outputs a line of length $k$, and identifies itself as the node $v_{jk+r_v+1}$. Finally,  the node transmits the message $(k,j,\delta_v)$. After receiving a message $(k,j,e)$ such that $e=(\delta_v+1) \mod 3$, a node $v$ of type 1 learns $k$ and $j$. Then it outputs a line of length $k$, and identifies itself as the node $v_{jk+1}$.

 The following lemma estimates the performance of Algorithm  {\tt Line-Topology-Recognition}.

\begin{lemma}\label{line ub}
Upon completion of Algorithm {\tt Line-Topology-Recognition} in a line of diameter at most $D$, where $D \geq 4$,
every node outputs the topology of the line and places itself in it correctly within time $ O(\log D )$.
\end{lemma}
\begin{proof}
It is enough to prove that every node computes the length $k$ of the line and its position in the line correctly.
Consider a node $v$ of type 1, and the closest type 0 node $u$ which is not a neighbor  of $v$.
Call the sequence of nodes starting at $v$ and ending at $u$ a {\em segment}.
According to the assignment of labels to the nodes in $T$, the concatenations of the second and third components of the labels of the nodes of the $j$th segment, are the binary representations of the  integers $k$ and $j$, respectively.

Let us consider the nodes in the $j$-th segment of the line $T$.
According to the formulation of Algorithm {\tt Line-Topology-Recognition}, the node $v$ of type 1 transmits the message $(0,\epsilon,\epsilon)$ in round 1. This message is received by two of its neighbors, one of which is the type 0 node of the $(j-1)$-th segment and the other is a type 2 node $w_1$ of the $j$-th segment. The type 0 node of the $(j-1)$-th segment ignores this message, as the strings in this message are empty strings. The node $w_1$ transmits the message $(0,\beta_{w_1},\gamma_{w_1})$ to the next  type 2 node $w_2$, which is  a neighbor of $w_1$. The node $w_2$ transmits the message $(0, (\beta_{w_1}\beta_{w_2}), (\gamma_{w_1}\gamma_{w_2}))$ to the next type 2 node $w_3$ and so on. If the segment contains any type 3 node, i.e., it is the last segment of the line $T$, then these nodes retransmit the message they received. Also, every node $w$ in this segment knows its distance from $v$ which it stored in the variable $r_w$. This is the round number when the node $w$ received the message for the first time from a node of this segment.
When the message $(0,s,s')$ reaches the unique type 0 node $u$  of this segment, then,  according to the labeling scheme, the strings $s$ and $s'$ are the binary representations of the integers $k$ and $j$, respectively. Therefore, the type 0 node of the $j$-th segment computes $k$ and $j$ correctly. Hence, the node $u$ of type 0 identifies its position as the node $v_{jk+r_u}$. Then $u$ transmits the message $(k,j,\delta_u)$. This message is received by the type 2 or type 3 neighbor  of  $u$ in the $j$-th segment, and the type 1 neighbor $v'$ of $u$ in the $(j+1)$-th segment. Since, $\delta_{v'}=(1+\delta_u) \mod 3$, the node $v'$ ignores this message. A type 2 or type 3 node $w'$ of the $j$-th segment learns $k$ and $j$ after receiving this message, identifies itself as the node $v_{jk+r_{w'}}$, and transmits the message $(k,j,\delta_{w'})$. When the message $(k,j,e)$ reaches the node $v$, i.e, the type 1 node of the $j$-th segment, $v$ identifies itself as the node $v_{jk+1}$. Therefore, every node in the $j$-th segment computes $k$ and identifies its position in the line of length $k$ correctly.

The algorithm starts when the type 1 node of a segment transmits the message $(0,\epsilon,\epsilon)$ and ends when every type 1 node in $T$ produces its output. Since the length of each segment is $O(\log D)$, the time complexity of the algorithm is $O(\log D)$.
\end{proof}

The following lemma gives a lower bound on the time of topology recognition for lines, matching the upper bound given  in Lemma \ref{line ub}.

\begin{lemma}\label {line lb}
Let $D \ge 3$ be any integer, and let $c>0$ be any real constant.
For any line $T$,  consider a labeling scheme LABEL($T$) of length at most $c$. Let $TOPO$ be any algorithm that solves topology recognition for every line of diameter at most $D$ using the labeling scheme LABEL($T$). Then  there exists a line of diameter at most $D$, for which $TOPO $ must take time $\Omega(\log D)$.
\end{lemma}
\begin{proof}
Let $T_j$, for $j=1,2\dots, D-\lfloor\frac{D}{2}\rfloor$, be the line of length $j+\lfloor\frac{D}{2}\rfloor$, with one endpoint $r$, considered as the root.
Let $\cT$ be the set of all lines $T_j$, $j=1,2\dots, D-\lfloor\frac{D}{2}\rfloor$. Then $|\cT| \ge \frac{D}{2}$.
Consider an algorithm {\it TOPO} that solves topology recognition for every tree $T\in \cT$ in time $\tau \le \frac{\log D}{2c+2}-1$, using a labeling scheme $LABEL(T)$ of length at most $c$.
In time $\tau$, at most $\tau$ nodes can reach the node $r$. Since there are at most $2^{c+1}$ different possible
labels of length at most $c$, the total number of possible histories $H(TOPO, \tau)$ of the root $r$ is at most
$2^{(c+1)(\tau+1)} \le  2^{\frac{\log D}{2}} <\frac{D}{2} <|\cT|$, for sufficiently large $D$.

Therefore, for sufficiently large $D$, there exist two trees $T'$ and $T''$ in $\cT$ such that the roots of the two trees have the same history. Hence the root $r$ in $T'$ and the root $r$ in $T''$ output the same tree as the topology. This is a contradiction, which proves the lemma.
\end{proof}

In view of  Lemmas \ref{line ub} and  \ref{line lb}, we have the following result.

\begin{theorem}
The optimal time for topology recognition in the class of trees of maximum degree $\Delta=2$ (i.e., lines) of  diameter at most $D$,
using a labeling scheme of length $O(1)$, is $\Theta(\log $D$)$.
\end{theorem}

\section{Conclusion}\label{s4}

We established a tight bound $\Theta(\log\log \Delta)$ on the minimum length of labeling schemes permitting topology recognition in trees of maximum degree $\Delta$, and we proved upper and lower bounds on topology recognition time, using such short schemes. These bounds on time are almost tight: they leave a multiplicative gap smaller than any polynomial in $\Delta$. Closing this small gap is a natural open problem. Another interesting research topic is to extend our results to the class of arbitrary graphs. We conjecture that such results, both concerning the minimum length of labeling schemes permitting topology recognition, and concerning the time necessary for this task, may be quite different from those that hold for trees.

%
%



\begin{thebibliography}{12}

\bibitem{AKM01}
S.~Abiteboul, H.~Kaplan, T.~Milo, Compact labeling schemes for ancestor
queries, Proc. 12th Annual ACM-SIAM Symposium on Discrete
Algorithms (SODA 2001), 547--556.

%


%





\bibitem{CGR}
M. Chrobak, L. Gasieniec, W. Rytter, Fast broadcasting and gossiping in radio networks,
Journal of Algorithms 43 (2002):177Ð189.


\bibitem{CFIKP}
R. Cohen, P. Fraigniaud, D. Ilcinkas, A. Korman, D. Peleg,
Label-guided graph exploration by a finite automaton,
ACM Transactions on Algorithms 4 (2008).

\bibitem{DP}
D. Dereniowski, A. Pelc, Drawing maps with advice,  Journal of Parallel and Distributed Computing 72 (2012), 132--143.



\bibitem{EFKR}
Y. Emek, P. Fraigniaud, A. Korman, A. Rosen, Online computation with advice, Theoretical Computer Science 412 (2011), 2642--2656.




\bibitem{FGIP}
P. Fraigniaud, C. Gavoille, D. Ilcinkas, A. Pelc,
Distributed computing with advice: Information sensitivity of graph coloring,
Distributed Computing 21 (2009), 395--403.

\bibitem{FIP1}
P. Fraigniaud, D. Ilcinkas, A. Pelc,
Communication algorithms with advice, Journal of  Computer and System Sciences 76 (2010), 222--232.

\bibitem{FIP2}
P. Fraigniaud, D. Ilcinkas, A. Pelc,
Tree exploration with advice, Information and Computation 206 (2008), 1276--1287.

\bibitem{FKL}
P. Fraigniaud, A. Korman, E. Lebhar,
Local MST computation with short advice,
Theory of Computing Systems 47 (2010), 920--933.


\bibitem{FP}
E. Fusco, A. Pelc, Trade-offs between the size of advice and broadcasting time in trees, Algorithmica 60 (2011), 719--734.


\bibitem{FPP}
E. Fusco, A. Pelc, R. Petreschi, Topology recognition with advice, Information and Computation 247 (2016), 254-265.

\bibitem{GPPR}
L. Gasieniec, A. Pagourtzis, I. Potapov, T. Radzik, Deterministic communication in radio networks with large labels. Algorithmica 47 (2007), 97-117.

\bibitem{GPX}
L. Gasieniec, D. Peleg, Q. Xin, Faster communication in known topology radio networks, Distributed Computing 19 (2007), 289-300.

\bibitem{GPPR02}
C.~Gavoille, D.~Peleg, S.~P\'{e}rennes, R.~Raz.
Distance labeling in graphs,
Journal of Algorithms 53 (2004), 85-112.

\bibitem{GMP}
C. Glacet, A. Miller, A. Pelc, Time vs. information tradeoffs for leader election in anonymous trees, Proc. 27th Annual ACM-SIAM Symposium on Discrete Algorithms (SODA 2016), 600-609.

%

\bibitem{IKP}
D. Ilcinkas, D. Kowalski, A. Pelc,
Fast radio broadcasting with advice,
 Theoretical Computer Science, 411 (2012),  1544--1557.

\bibitem{KKKP02}
M.~Katz, N.~Katz, A.~Korman, D.~Peleg, Labeling schemes for flow and
connectivity,
SIAM Journal of  Computing 34 (2004), 23--40.


\bibitem{KKP05}
A. Korman, S. Kutten, D. Peleg, Proof labeling schemes,
Distributed Computing 22 (2010), 215--233.

\bibitem{KP}
D. Kowalski, A. Pelc, Leader election in ad hoc radio networks: a keen ear helps, Journal of Computer and System Sciences 79 (2013), 1164-1180.





\bibitem{SN}
N. Nisse, D. Soguet, Graph searching with advice,
Theoretical Computer Science 410 (2009), 1307--1318.

 \bibitem{Pe}D. Peleg,
  Distributed computing, a locality-sensitive approach,
  SIAM Monographs on Discrete Mathematics and Applications, Philadelphia 2000.



%


\end{thebibliography}
\end{document}